\documentclass{article}

\usepackage{amsthm,amsmath,amssymb,latexsym,graphicx,enumerate}
\usepackage[mathscr]{eucal}
\usepackage{color}
\usepackage[pdfencoding=auto, psdextra]{hyperref}

\newtheorem{definition}{Definition}
\newtheorem{remark}{Remark}
\newtheorem{lemma}{Lemma}
\newtheorem{example}{Example}
\newtheorem{theorem}{Theorem}
\newtheorem{corollary}{Corollary}

\newcommand{\polylog}{\mathrm{PolyLogTime}}
\newcommand{\npolylog}{\mathrm{NPolyLogTime}}

\begin{document}

\title{A Restricted Second-Order Logic for Non-deterministic Poly-Logarithmic Time\footnote{Draft of Paper submitted to the Logic Journal of the IGPL.}}

\author{Flavio Ferrarotti\\
Software Competence Center Hagenberg, Austria \\
Email: flavio.ferrarotti@scch.at\\
\and
Senen Gonz\'{a}les\\
P\&T Connected, Hagenberg, Austria \\
Email: ulcango@gmail.com\\
\and
Klaus-Dieter Schewe\\
Zhejiang University, UIUC Institute, Haining, China\\
Email: kd.schewe@intl.zju.edu.cn\\
\and
Jos\'{e} Mar\'{i}a Turull-Torres\\
Universidad Nacional de La Matanza, Buenos Aires, Argentina\\
Email: jturull@unlam.edu.ar}

\date{February 2019}


\maketitle

\begin{abstract}
We introduce a restricted second-order logic $\mathrm{SO}^{\mathit{plog}}$ for finite structures where second-order quantification ranges over relations of size at most poly-logarithmic in the size of the structure. 
We demonstrate the relevance of this logic and complexity class by several problems in database theory.
We then prove a Fagin's style theorem showing that the Boolean queries which can be expressed in the existential fragment of $\mathrm{SO}^{\mathit{plog}}$ corresponds exactly to the class of decision problems that can be computed by a non-deterministic Turing machine with random access to the input in time $O((\log n)^k)$ for some $k \ge 0$, i.e., to the class of problems computable in non-deterministic poly-logarithmic time. It should be noted that unlike Fagin's theorem which proves that the existential fragment of second-order logic captures NP over arbitrary finite structures, our result only holds over ordered finite structures, since $\mathrm{SO}^{\mathit{plog}}$ is too weak as to define a total order of the domain. Nevertheless $\mathrm{SO}^{\mathit{plog}}$ provides natural levels of expressibility within poly-logarithmic space in a way which is closely related to how second-order logic provides natural levels of expressibility within polynomial space. Indeed, we show an exact correspondence between the quantifier prefix classes of $\mathrm{SO}^{\mathit{plog}}$ and the levels of the non-deterministic poly-logarithmic time hierarchy, analogous to the correspondence between the quantifier prefix classes of second-order logic and the polynomial-time hierarchy. Our work closely relates to the constant depth quasipolynomial size AND/OR circuits and corresponding restricted second-order logic defined by David A. Mix Barrington in 1992. We explore this relationship in detail. 
\\[1ex]

\noindent
{\normalsize\textbf{Acknowledgements.}} The research reported in this paper results from the project {\em Higher-Order Logics and Structures} supported by the Austrian Science Fund (FWF: \textbf{[I2420-N31]}).
It was further supported by the Austrian Research Promotion Agency (FFG) through the COMET funding for the Software Competence Center Hagenberg. 
\end{abstract}

\section{Introduction}

During the last forty years logics over finite structures have become a central pillar for studying the definability and complexity of computational problems. The focus is on understanding how the expressive power of logics over finite structures, or equivalently query languages over relational databases, relate to natural classes of computational complexity. The foundational result in this line of work is Fagin's famous theorem~\cite{fagin:1973} which states that the existential fragment $\mathrm{SO}\exists$ of second-order logic over finite relational structures captures all decision problems that are accepted by a non-deterministic Turing machine in polynomial time --in other words: $\mathrm{SO}\exists$ captures the complexity class $\mathrm{NP}$. This was extended by Stockmeyer~\cite{Stockmeyer76} to an exact correspondence between the quantifier prefix classes of second-order logic and the levels of the polynomial-time hierarchy. Since then, most Turing complexity classes have been characterized in terms of the expressive power of logic languages (see e.g. the monographs by Immerman \cite{Immerman99} and Libkin \cite{Libkin04} or the collection \cite{graedel:eatcs2007}). 

An advantage of second-order logic is that it provides a natural and high level of expressive power. A simple example which illustrates this point is provided by the set-containment join. A query that asks whether a patient has all the symptoms associated to a given disease, can be written literally provided the query language has the appropriate second-order constructs. By contrast, in first-order logic we cannot write this literally. We would need to say instead that for all symptom $s$, if the patient $p$ has $s$, then $s$ is also a symptom of the disease. Unfortunately, high expressiveness of second-order logic also yields a high complexity of evaluation of formulae as shown by Fagin-Stockmeyer theorems, which in principle make them not suitable for practical purposes. 
Nevertheless, second-order logic has been used in applied areas such as Knowledge Representation~\cite{BruynoogheB0CPJ15}. In that area it is usually known as ``model expansion for first-order logic'' and SAT solving is used to find the existentially quantified relations. On the other hand, the SAT solvers are usually ``helped'' by adding explicit syntax for fixed-points, both least and greatest. Also nested fixed-points (simultaneous induction) have been used for this purpose~\cite{HouCD10}.

Aiming at a better understanding of which features of second-order logic have a real impact on its expressive power and complexity, several semantic and syntactic restrictions have been considered in the literature. Among the syntactic restrictions, the results in \cite{Gradel92} should be highlighted. The logics SO-Horn and SO-Krom obtained by restricting the second-order logic to Horn and Krom formulae, respectively, both collapse to their respective existential fragments. Moreover, in finite structures that include the successor relation, they provide characterizations of deterministic polynomial-time and nondeterministic logspace, respectively. Also
the tractability/intractability frontier of the model checking problem for prefix classes of
existential second-order logic has been completely delineated (see~\cite{EiterGS01,EiterGG00,GottlobKS04}).

Regarding semantic restrictions of second-order logic, the logic $\mathrm{SO}^\omega$ introduced by A. Dawar
in~\cite{Dawar98} and the related logic $\mathrm{SO}^F$ introduced in~\cite{GrossoT10} are the source of inspiration for this paper. Both logics restrict the interpretation of second-order quantifiers to relations closed under equivalence of types of the tuples in the given relational structure. In the case of $\mathrm{SO}^\omega$, the second-order quantification is restricted to relations closed under equivalence of $\mathrm{FO}^k$-types of the tuples, where $\mathrm{FO}^k $ is the restriction of first-order logic to formulae with at most $k$ different variables. In $\mathrm{SO}^F$ the quantification is restricted to relations closed under equivalence of first-order
types of the tuples, i.e., under isomorphic types. It was proven in~\cite{Dawar98}, among other results, that the expressive power of the existential fragment of $\mathrm{SO}^\omega$ is equivalent to the expressive power of the nondeterministic inflationary fixed-point logic, and thus that $\mathrm{SO}^\omega$ is contained within the infinitary logic with finitely many variables ${\cal L}^\omega_{\infty, \omega}$. As shown in~\cite{GrossoT10}, $\mathrm{SO}^F$ is strictly more expressive than $\mathrm{SO}^{\omega}$. In the absence of linear order many natural NP-complete problems such as Hamiltonicity and clique are not expressible $\mathrm{SO}^{\omega}$ since they are already not expressible in ${\cal L}^\omega_{\infty, \omega}$ (see~\cite{Immerman99} among other sources). On the other hand it is easy to
see that there are NP-complete problems that can be expressed in the existential fragment of
$\mathrm{SO}^\omega$, since this logic captures NP on ordered structures. 
Through the study of different semantic restrictions over binary NP, i.e., existential second-order logic with second-order quantification restricted to binary relations, many interesting results regarding the properties of the class of problems expressible in this logic were established~\cite{DurandLS98}. Semantic restrictions were based mainly on second-order quantification restricted to unary functions, order relations and graphs with degree bounds. Based on these restrictions they were able to prove the existence of a strict hierarchy of binary NP problems.
Another relevant example of a semantic restriction over existential second-order logic can
be found in~\cite{LautemannST94}. It was shown in that work that context-free languages coincide with the
class of those sets of strings that can be defined, on word models, by existential second-order
sentences in which the second-order quantifiers range over a restricted class of binary relations
called matchings.

So the question comes up: Are there additional semantic restrictions of second-order logic that can result in elegant descriptive characterizations of meaningful computational complexity classes? It turns out that the approach of simply restricting second-order quantification to range over relations of size at most poly-logarithmic in the size of the structure, already leads to a positive answer. Indeed, using this approach we define a restricted second-order logic, namely $\mathrm{SO}^{\mathit{plog}}$, and prove a Fagin's style theorem showing that Boolean queries which can be expressed in the existential fragment of $\mathrm{SO}^{\mathit{plog}}$ corresponds exactly to the class of decision problems that can be computed by a non-deterministic Turing machine with random access to the input in time $O((\log n)^k)$ for some $k \ge 0$, i.e., to the class of problems computable in non-deterministic poly-logarithmic time ($\npolylog$ for short). It should be noted that unlike Fagin's theorem which proves that the existential fragment of second-order logic captures NP over arbitrary finite structures, our result only holds over ordered finite structures, since $\mathrm{SO}^{\mathit{plog}}$ is too weak as to define a total order of the domain. Nevertheless $\mathrm{SO}^{\mathit{plog}}$ provides natural levels of expressibility within poly-logarithmic space in a way which is closely related to how second-order logic provides natural levels of expressibility within polynomial space. In fact, we show an exact correspondence between the expressive power of the quantifier prefix classes of $\mathrm{SO}^{\mathit{plog}}$ and the levels of the non-deterministic poly-logarithmic time hierarchy (polylog-time hierarchy from now on), analogous to the correspondence between the quantifier prefix classes of second-order logic and the polynomial-time hierarchy.

This is up to our knowledge the first descriptive characterization of $\npolylog$ and each subsequent level of the polylog-time hierarchy. An anonymous referee of the preliminary conference version of the current paper~\cite{FerrarottiGST08}, pointed us however to a very relevant antecedent in the work of David A. Mix Barrington in~\cite{Barr92}, where a semantically restricted second-order logic (let us denote it as $\mathrm{SO}^b$) related to our logic $\mathrm{SO}^{\mathit{plog}}$, is used to characterize a class of families of constant depth quasipolynomial size AND/OR circuits $\mathit{qAC}^0$. In particular it is shown there that the class of Boolean queries computable by $\mathrm{DTIME}[(\log n)^{{O}(1)}]$ $\mathrm{DCL}$-uniform families of Boolean circuits of unbounded fan-in, size $2^{({\log n})^{{O}(1)}}$ and depth ${O}(1)$, coincides with the class of Boolean queries expressible in $\mathrm{SO}^{b}$. While this would imply that $\mathrm{SO}^{b}$ also captures the whole polylog-time hierarchy (see Section~\ref{barrington} for a detailed explanation), in the case of $\mathrm{SO}^{\mathit{plog}}$ this is an easy corollary of the one-to-one correspondence between its quantifier prefix classes and the levels of the polylog-time hierarchy. As we show in Section~\ref{barrington},  this correspondence is very unlikely to hold for the quantifier prefix classes of $\mathrm{SO}^b$. It is also very unlikely that the existential fragment of $\mathrm{SO}^b$ can provide a descriptive characterization of $\npolylog$, as it appears to be too powerful for that.

We further believe that the natural levels of expressive power provided by $\mathrm{SO}^{\mathit{plog}}$ are not matched by $\mathrm{SO}^b$. In this sense, we give examples of natural queries expressible in $\mathrm{SO}^{\mathit{plog}}$, such as the classes $\mathrm{DNFSAT}$ of satisfiable propositional formulas in disjunctive normal form and $\mathrm{CNFTAUT}$ of propositional tautologies in conjunctive normal form, both defined in as early as 1971 (\cite{Cook_71}). The definition of such queries in $\mathrm{SO}^{\mathit{plog}}$ can be done by means of relatively simple and elegant formulae, despite a restriction we need to impose in the universal first-order quantification. This is not fortuitous, but the consequence of the fact that in the definition of $\mathrm{SO}^{\mathit{plog}}$ we use a more relaxed notion of second-order quantification than that used in the definition of $\mathrm{SO}^b$. Indeed, the second-order quantifiers in $\mathrm{SO}^{\mathit{plog}}$ range over arbitrary relations of polylog size on the number of elements of the domain, not just over relations defined on the set formed by the first $\log n$ elements of that domain as in $\mathrm{SO}^b$. The descriptive complexity of $\mathrm{SO}^{\mathit{plog}}$ is not increased by this more liberal definition of polylog restricted second-order quantifiers.

We reach our results by following an inductive itinerary. After presenting some short but necessary preliminaries in Section \ref{sec:preliminaries}, we introduce the logic $\mathrm{SO}^{\mathit{plog}}$ in Section~\ref{sec:soplog}. We do this in a comprehensive way, giving examples of problems expressible in $\mathrm{SO}^{\mathit{plog}}$. The fragments $\Sigma^{\mathit{plog}}_m$ and $\Pi^{\mathit{plog}}_m$ of formulae in quantifier prenex normal form are defined using the classical approach in second-order logic, showing that every $\mathrm{SO}^{\mathit{plog}}$ formula can be written in this normal form. This forms the basis for the definition of the hierarchy inside $\mathrm{SO}^{\mathit{plog}}$. 

Section~\ref{sec:arithexampls} shows how the first level $\Sigma^{\mathit{plog}}_1$ of the hierarchy of quantifier prenex formulae of $\mathrm{SO}^{\mathit{plog}}$ can already define the (poly-logarithmically) bounded binary arithmetics necessary to prove our main result, i.e., to prove that the existential fragment of $\mathrm{SO}^{\mathit{plog}}$ captures $\npolylog$. We should stress that it is \emph{not} immediately obvious that these operations can be expressed in $\Sigma^{\mathit{plog}}_1$, since this logic cannot express all existential second-order properties over relations of polylogarithmic size due to its restricted universal first-order quantification.

In Section~\ref{sec:plh} we concentrate on complexity classes inside polylogarithmic space. Analogous to the polynomial time hierarchy inside polynomial space we define a polylog-time hierarchy $\mathrm{PLH}$, where $\tilde{\Sigma}_1^{\mathit{plog}}$ is defined by $\npolylog$ capturing all decision problems that can be accepted by a non-deterministic Turing machine in time $O((\log n)^k)$ for some $k \ge 0$, where $n$ is the size of the input. In order to be able to deal with the sublinear time constraint we assume random access to the input following the same approach than in~\cite{barrington:jcss1990}. Higher complexity classes $\tilde{\Sigma}_m^{\mathit{plog}}$ (and $\tilde{\Pi}_m^{\mathit{plog}}$) in the hierarchy are defined in a similar way using alternating Turing machines with a bound $m$ on the alternations.

Section \ref{sec:main} contains our main results. First we give a detailed, constructive proof of the fact that the existential fragment of $\mathrm{SO}^{\mathit{plog}}$, i.e. $\Sigma^{\mathit{plog}}_1$, captures the complexity class $\npolylog$. After that, we follow the inductive path and establish the expressive power of the fragments $\Sigma^{\mathit{plog}}_m$ and $\Pi^{\mathit{plog}}_m$, for every $m \geq 1$, proving that each layer is characterized by a random-access alternating Turing machine with polylog time and $m$ alternations. The fact that $\mathrm{PLH} = \mathrm{SO}^{\mathit{plog}}$ follows as a simple corollary.

The way in which the restricted second-order logic and corresponding class of families of circuits $\mathit{qAC}^0$ studied by David A. Mix Barrington in~\cite{Barr92} relates to our work is investigated formally in Section~\ref{barrington}. We conclude the paper with a brief summary and outlook in Section \ref{sec:schluss}.

\section{Preliminaries}\label{sec:preliminaries}

Unless otherwise stated, we work with ordered finite structures and assume that all vocabularies include the relation and constant symbols: $\leq$, $\mathrm{SUCC}$, $\mathrm{BIT}$, $0$, $1$, $\mathit{logn}$ and $\mathit{max}$. In every structure $\bf A$, $\leq$ is interpreted as a total ordering of the domain $A$ and $\mathrm{SUCC}$ is interpreted by the successor relation corresponding to the $\leq^{\bf A}$ ordering. The constant symbols $0$, $1$ and $\mathit{max}$ are in turn interpreted as the minimum, second and maximum elements under the $\leq^{\bf A}$ ordering and the constant $\mathit{logn}$ as $\left\lceil \log_2 |A| \right\rceil$. By passing to an isomorphic copy, we assume that $A$ is the set $\{0, 1, \ldots, n-1\}$ of natural numbers less than $n$, where $n$ is the cardinality $|A|$ of $A$. Then $\mathrm{BIT}$ is interpreted by the following binary relation:
\[\mathrm{BIT}^{\bf A} = \{(i, j) \in A^2 \mid \text{Bit $j$ in the binary representation of $i$ is $1$}\}.\]
In this paper, $\log n$ always refers to the binary logarithm of $n$, i.e. $\log_2 n$. We write $\log^k n$ as a shorthand for $(\left\lceil\log n \right\rceil)^k$ and finally $\mathit{\log n-1}$ as $z$ such as $SUCC(z,\mathit{logn})$. 
We assume that all structures have at least \emph{three} elements. This results in a cleaner presentation, avoiding the trivial cases of structures with only one element which would satisfy $0 = 1$ and structures with only two elements which would unnecessarily complicate the definition of the bounded binary arithmetic operations in Section~\ref{sec:arithexampls}.

\section{\texorpdfstring{$\mathrm{SO}^{\mathit{plog}}$}{TEXT}: A Restricted Second-Order Logic}{\label{sec:soplog}}

We define $\mathrm{SO}^{\mathit{plog}}$ as the restricted second-order logic obtained by extending \emph{existential} first-order logic with (1) universal and existential second-order quantifiers that are restricted to range over relations of poly-logarithmic size in the size of the structure, and (2) universal first-order quantifiers that are restricted to range over the tuples of such poly-logarithmic size relations. 

\begin{definition}[Syntax of $\mathrm{SO}^{\mathit{plog}}$]
For every $r{\geq}1$ and $k{\geq}0$, the language of $\mathrm{SO}^{\mathit{plog}}$ extends the language of first-order logic with countably many second-order variables $X_1^{r,\log^k}$, $X_2^{r,\log^k}, \dots$ of {\em arity $r$} and {\em exponent $k$}. The set of well-formed $\mathrm{SO}^{\mathit{plog}}$-formulae (wff) of vocabulary $\sigma$ is inductively defined as follows:
\begin{enumerate}[i.]

\item Every wff of vocabulary $\sigma$ in the existential fragment of first-order logic with equality is a wff.

\item If $X^{r,\log^k}$ is a second-order variable and $t_1, \ldots, t_r$ are first-order terms, then both $X^{r,\log^k}(t_1, \ldots, t_r)$ and $\neg X^{r,\log^k}(t_1, \ldots, t_r)$ are wff's.

\item If $\varphi$ and $\psi$ are wff's, then $(\varphi \wedge \psi)$ and $(\varphi \vee \psi)$ are wff's.
    
\item If $\varphi$ is a wff, $X^{r,\log^k}$ is a second-order variable and $\bar{x}$ is an $r$-tuple of first-order variables, then $\forall \bar{x} (X^{r,\log^k}(\bar{x}) \rightarrow \varphi)$ is a wff.    
    
\item If $\varphi$ is a wff and $x$ is a first-order variable, then $\exists x \varphi$ is a wff.   
     
\item If $\varphi$ is a wff and $X^{r,\log^k}$ is a second-order variable, then both $\exists X^{r,\log^k} \varphi$ and $\forall X^{r,\log^k} \varphi$ are wff's.
    
\end{enumerate}
\end{definition}

Note that the first-order terms $t_i$ in these rules are either first-order variables $x_1, x_2, \ldots$ or constant symbols; we do not consider function symbols. Whenever the arity is clear from the context, we write $X^{\log^k}$ instead of $X^{r,\log^k}$. 
    
\begin{definition}[Semantics of $\mathrm{SO}^{\mathit{plog}}$]
Let $\mathbf{A}$ be a $\sigma$-structure where $|A| = n \geq 2$. A valuation over $\mathbf{A}$ is any function \textit{val} which assigns appropriate values to all first- and second-order variables and satisfies the following constraints: 
\begin{itemize}

\item If $x$ is a first-order variable then $\mathit{val}(x) \in A$. 

\item If $X^{r,\log^k}$ is a second-order variable, then 
\[\mathit{val}(X^{r,\log^k}) \in \{R\subseteq A^r \mid |R| \leq (\lceil \log n \rceil)^k\}.\]
\end{itemize}
\end{definition}
As usual, we say that a valuation $\mathit{val}$ is $V$-equivalent to a valuation $\mathit{val}'$ if $\mathit{val}(V') = \mathit{val}'(V')$ for all variables $V'$ other than $V$.

$\mathrm{SO}^{\mathit{plog}}$ extends the notion of satisfaction of first-order logic, with the following rules:
\begin{itemize}
     \item $\mathbf{A},\mathit{val} \models X^{r,\log^k}(x_1,\dots,x_r) $ iff $(\mathit{val}(x_1),\dots,\mathit{val}(x_r))\in \mathit{val}( X^{r,\log^k})$.
     \item $\mathbf{A},\mathit{val} \models \neg X^{r,\log^k}(x_1,\dots,x_r) $ iff $(\mathit{val}(x_1),\dots,\mathit{val}(x_r))\not\in \mathit{val}( X^{r,\log^k})$.
     
     \item $\mathbf{A},\mathit{val} \models \exists X^{r,\log^k} (\varphi)$  iff there is a valuation $\mathit{val}'$ which is $X^{r,\log^k}$-equivalent to $\mathit{val}$ such that $\mathbf{A}, \mathit{val}' \models \varphi$.
     
     \item $\mathbf{A},\mathit{val} \models \forall X^{r,\log^k} (\varphi)$  iff, for all valuations $\mathit{val}'$ which are $X^{r,\log^k}$-equivalent to $\mathit{val}$, it holds that $\mathbf{A}, \mathit{val}' \models \varphi$.
     
\end{itemize}

\begin{remark}\label{r1}
The standard (unbounded) universal quantification of first-order logic formulae of the form 
$\forall x \varphi$ can be expressed in $\mathrm{SO}^{\mathit{plog}}$ by formulae of the form $\forall X^{\log^0} \forall x (X^{\log^0}(x) \rightarrow \varphi)$. Thus, even though $\mathrm{SO}^{\mathit{plog}}$ only allows a restricted form of universal first-order quantification, it can nevertheless express every first-order query. This is however not applicable to its existential fragment.   
\end{remark} 

We denote by $\Sigma^{\mathit{plog}}_m$, where $m \geq 1$, the class of $\mathrm{SO}^{\mathit{plog}}$-formulae of the form:
\[\exists X^{\log^{k_{11}}}_{11} \cdots \exists X^{\log^{k_{1s_1}}}_{1s_1} \forall X^{\log^{k_{21}}}_{21} \cdots \forall X^{\log^{k_{2s_2}}}_{2s_2} \cdots Q X^{\log^{k_{m1}}}_{m1} \cdots Q X^{\log^{k_{ms_m}}}_{ms_m} \psi,\]
where $Q$ is either $\exists$ or $\forall$ depending on whether $m$ odd or even, respectively, and $\psi$ is an $\mathrm{SO}^{\mathit{plog}}$-formula \emph{free} of second-order quantifiers.  
Analogously, we denote by $\Pi^{\mathit{plog}}_m$ the class of $\mathrm{SO}^{\mathit{plog}}$-formulae of the form:
\[\forall X^{\log^{k_{11}}}_{11} \cdots \forall X^{\log^{k_{1s_1}}}_{1s_1} \exists X^{\log^{k_{21}}}_{21} \cdots \exists X^{\log^{k_{2s_2}}}_{2s_2} \cdots Q X^{\log^{k_{m1}}}_{m1} \cdots Q X^{\log^{k_{ms_m}}}_{ms_m} \psi.\]

We say that an $\mathrm{SO}^{\mathit{plog}}$-formula is in \emph{quantifier prefix normal form} (QNF) if it belongs to either $\Sigma^{\mathit{plog}}_m$ or $\Pi^{\mathit{plog}}_m$ for some $m \geq 1$.

\begin{lemma}\label{lem-snf}
For every $\mathrm{SO}^{\mathit{plog}}$-formula $\varphi$, there is an equivalent $\mathrm{SO}^{\mathit{plog}}$-formula $\varphi'$ that is in QNF.
\end{lemma} 

\begin{proof}
An easy induction using renaming of variables and equivalences such as $(\neg \exists X^{\log^k} \varphi)$ $\equiv \forall X^{\log^k} (\neg \varphi)$ and $(\phi \vee
\exists x \psi) \equiv \exists x (\phi \vee \psi)$ if $x$ is not free in $\phi$, shows that each $\mathrm{SO}^{\mathit{plog}}$-formula is logically equivalent to an $\mathrm{SO}^{\mathit{plog}}$-formula in \emph{prenex normal form}, i.e., to a formula where all first- and second-order quantifiers are grouped together at the front, forming alternating blocks of consecutive existential or universal quantifiers. Yet the problem is that first- and second-order quantifiers might be mixed. Among the quantifiers of a same block, though, it is clearly possible to commute them 
so as to get those of second-order at the beginning of the block. But, we certainly cannot commute different quantifiers without altering the meaning of the formula. What we can do is to replace first-order quantifiers by second-order quantifiers so that all quantifiers at the beginning of the formula are of second-order, and they are then eventually followed by first-order quantifiers. This can be done using the following equivalences:
\[ \exists x \forall Y^{\log^k} \psi \equiv \exists X^{\log^0} \forall Y^{\log^k} \exists x (X^{\log^0}(x) \wedge \psi).\]
\[ \forall x \exists Y^{\log^k} \psi \equiv \forall X^{\log^0} \exists Y^{\log^k} \forall x (X^{\log^0}(x) \rightarrow \psi). \]
\end{proof}

Next we present examples of problems which are definable in $\mathrm{SO}^{\mathit{plog}}$. We start with a simple but useful example, and then move to examples which give a better idea of the actual expressive power of $\mathrm{SO}^{\mathit{plog}}$.

\begin{example}
Let $X$ and $Y$ be $\mathrm{SO}^{\mathit{plog}}$ variables of the form $X^{r_1,\log^{k}}$ and $Y^{r_2,\log^{k}}$. The following $\Sigma^{\mathit{plog}}_1$ formula, denoted as $|X| \leq |Y|$, expresses that the cardinality of (the relation assigned by the current valuation of) $X$ is less than or equal to that of $Y$. \\[0.1cm]
\hspace*{0.38cm}$\exists R \Big(\forall \bar{x} \big(X(\bar{x}) \to \exists \bar{y} \big( Y(\bar{y})\wedge R(\bar{x},\bar{y}) \wedge \forall \bar{z} (X(\bar{z})\to (\bar{z}{\neq} \bar{x} \to \neg R(\bar{z},\bar{y})))\big)\big)\Big)$,\\[0.1cm]
where $R$ is an $\mathrm{SO}^{\mathit{plog}}$ variable of arity $r_1+r_2$ and exponent $k$. In turn $|X|{=}|Y|$ can be defined as $|X|{\leq}|Y| \wedge |Y|{\leq}|X|$.
\end{example}

It is not difficult to see that existential $\mathrm{SO}^{\mathit{plog}}$ can naturally define what we could call poly-logarithmically bounded versions of NP complete problems.

\begin{example}\label{polylogclique}
Let $G=(V,E)$ be an $n$-node undirected graph. The following sentence expresses a poly-logarithmically bounded version of the clique NP-complete problem. It holds iff $G$ contains a clique of size $\lceil \log n\rceil^k$. 
\[\exists I S \big(\mathrm{DEF}_k(I) \wedge |S|{=}|I| \wedge\forall x \big(S(x) {\to} \forall y (S(y) \to (x\neq y \to (E(x,y) \wedge E(y,x))))\big)\big)\]
Here $I$ and $S$ are second-order variables of arity and exponent $k$ and $\mathrm{DEF}_k(I)$ holds iff $I$ is interpreted with the $k$-ary relation $\{0, \ldots, \lceil \log n \rceil -1\}^k$. Clearly $\mathrm{DEF}_k(I)$ can be defined in existential $\mathrm{SO}^{\mathit{plog}}$ as shown in~(\ref{arith3}) in our next section. 
Other bounded versions of classical Boolean NP-complete problems that are easily expressible in $\Sigma^{\mathit{plog}}_1$ are for instance to decide whether $G$ has an induced subgraph of size $\lceil \log n \rceil^k$ that is $3$-colourable, or whether a $G$ has an induced subgraph which is isomorphic to another given graph of at most polylog size w.r.t. the size of $G$.
\end{example}

We conclude this section with an example of a $\mathrm{SO}^{\mathit{plog}}$ sentence which expresses the standard version of DNFSAT.
\begin{example}
Let DNFSAT denote the class of satisfiable propositional formulas in disjunctive normal form.
In the standard encoding of DNF formulae as word models of alphabet $\sigma = \{(,),\wedge,\vee, \neg, 0, 1, X\}$, DNFSAT is decidable in $P$~\cite{Cook_71}. In this encoding, 
the input formula is a disjunction of arbitrarily many clauses enclosed in pairs of matching parenthesis. Each clause is the conjunction of an arbitrary number of literals. Each literal is a variable of the form $X_w$, where the subindex $w \in \{0, 1\}^*$, possibly preceded by a negation symbol.
Obviously, the complement NODNFSAT of DNFSAT is also in P. In $\Pi^{\mathrm{plog}}_2$ we can define NODNFSAT by means of a sentence stating that for every clause there is a pair of complementary literals. Every clause is logically defined by a pair of matching parentheses such that there is no parenthesis in between. A pair of complementary literals is defined by a bijection (of size $\leq \lceil \log n \rceil$) between the subindexes of two literals, which preserves the bit values and such that exactly one of the literals is negated. The following sentence expresses this formally. \\[0.4cm]
$\forall x_0 x_1 \Bigg(I_((x_0) \wedge I_)(x_1) \wedge \forall x \big(x_0 < x < x_1 \rightarrow \neg (I_((x)\vee I_)(x)\big) \rightarrow$\\ 
\hspace*{1.3cm} $\exists H \bigg(\forall x y \big(H(x, y) \rightarrow x_0 < x < x_1 \wedge x_0 < y < x_1 \big) \wedge$\\ 
\hspace*{1.9cm} $\forall x y z \big(H(x, z) \wedge H(y, z) \rightarrow x = y\big) \wedge$\\[0.1cm]
\hspace*{1.9cm} $\forall x y z \big(H(x, y) \wedge H(x, z) \rightarrow y = z\big) \wedge$\\[0.1cm]
\hspace*{1.9cm} $\forall x y \big(H(x, y) \rightarrow \big((I_0(x) \wedge I_0(y)) \vee (I_1(x) \wedge I_1(y)) \big) \big) \wedge$\\[0.1cm]
\hspace*{1.9cm} $\exists x_2 x_3 x_4 x_5 y_2 y_3 y_4 y_5 \Big(\mathrm{SUCC}(x_2, x_3) \wedge \mathrm{SUCC}(y_2, y_3) \wedge$\\
\hspace*{2.5cm} $(I_((x_2) \vee I_\wedge(x_2)) \wedge I_X(x_3) \wedge (I_\wedge(x_5) \vee I_)(x_5)) \wedge$\\[0.1cm]
\hspace*{2.5cm} $I_\neg(y_2) \wedge I_X(y_3) \wedge (I_\wedge(y_5) \vee I_)(y_5)) \wedge $ \\[0.1cm]
\hspace*{2.5cm} $\mathrm{SUCC}(x_3, x_4) \wedge \mathrm{SUCC}(y_3, y_4) \wedge H(x_4, y_4) \wedge $\\[0.1cm]
\hspace*{2.5cm} $\forall x y \big(H(x, y) \rightarrow \big(\exists z_1 z_2 (\mathrm{SUCC}(x, z_1) \wedge \mathrm{SUCC}(y, z_2) \wedge H(z_1, z_2) ) \vee$\\[0.1cm] 
\hspace*{5.1cm} $(\mathrm{SUCC}(x, x_5) \wedge \mathrm{SUCC}(y, y_5))\big)\big)\Big)\bigg)\Bigg)$\\[0.1cm]
It is fairly easy to see that this formula can be translated into an equivalent formula in $\Pi^{\mathit{plog}}_2$. Note that the (unbounded) first-order universal quantifiers in the first line can be replaced by $\mathrm{SO}^{\mathit{plog}}$ quantifiers of second-order with variables of exponent $1$ as per Remark~\ref{r1}. The exponent of $H$ is $1$ as well. 

Similarly, DNFSAT can be defined in $\Sigma^{\mathit{plog}}_2$ by a sentence stating that there is a clause that does not have a pair of complementary literals.
\end{example}

\section{Bounded Binary Arithmetic Operations in \texorpdfstring{$\Sigma^{\mathit{plog}}_1$}{TEXT}}\label{sec:arithexampls}\label{sec:examples}

We define $\Sigma^{\mathit{plog}}_1$-formulae that describe the basic (bounded) arithmetic operations of sum, multiplication, division and modulo among binary positive integers between $0$ and $2^{\lceil\log n \rceil^k}-1$ for some fixed $k \geq 1$. These formulae are later needed for proving our main result regarding the expressive power of the existential fragment of $\mathrm{SO}^{\mathit{plog}}$.  

Note that it is \emph{not} immediately obvious that these operations can indeed be expressed in $\Sigma^{\mathit{plog}}_1$, since this logic cannot express all existential second-order properties over relations of polylogarithmic size due to its restricted universal first-order quantification.

In our approach, binary numbers between $0$ and $\lceil 2^{(\log n)^k} \rceil - 1$ are represented by means of ($\mathrm{SO}^{\mathit{plog}}$) relations. 
\begin{definition}\label{def:repBinaryNumbers}
Let $b = b_0 \cdots b_l$ be a binary number, where $b_0$ and $b_l$ are the least and most significant bits of $b$, respectively, and $l \leq \lceil \log n \rceil^k$. Let $B = \{0, \ldots, \lceil \log n \rceil -1\}$. The relation $R_b$ \emph{encodes the binary number} $b$ if the following holds: $(a_0, \ldots, a_{k-1}, a_k) \in R_b$ iff $(a_0, \ldots, a_{k-1}) \in B^k$ is the $i$-th tuple in the lexicographical order of $B^k$, $a_k = 0$ if $i > l$, and $a_k = b_i$ if $0 \leq i \leq l$. 
\end{definition}

Note that the size of $R_b$ is exactly $\lceil \log n \rceil^k$, and thus $R_b$ is a valid valuation for $\mathrm{SO}^{\mathit{plog}}$ variables of the form $X^{k+1, \log^k}$.
The numerical order relation $\leq_k$ among $k$-tuples can be defined as follows: 
\begin{align}
x_0 \leq_1 y_0 \quad &\equiv \quad x_0 \leq y_0 \quad \text{and} \notag\\
\bar{x} \leq_k \bar{y} \quad &\equiv \quad (x_0 \leq y_0 \wedge x_0 \neq y_0) \vee (x_0 = y_0 \wedge (x_1, \ldots, x_{k-1}) \leq_{k-1} (y_1, \ldots, y_{k-1})) \label{arith1}
\end{align}

In our approach, we need a successor relation $\mathrm{SUCC}_k$ among the $k$-tuples in $B^k$, where $B$ is the set of integers between $0$ and $\lceil \log n \rceil-1$ (cf. Definition~\ref{def:repBinaryNumbers}). 
\begin{align}
  \mathrm{SUCC}_1(x_0,y_0) &\equiv y_0 \leq \mathit{logn} \wedge y_0 \neq \mathit{logn} \wedge \mathrm{SUCC}(x_0,y_0) \quad \text{and} \notag\\
\mathrm{SUCC}_k(\bar{x},\bar{y}) &\equiv y_0 \leq \mathit{logn} \wedge y_0 \neq \mathit{logn} \wedge \notag \\
[\, (y_0 = x_0 \wedge & \mathrm{SUCC}_{k-1}(x_1,\dots,x_{k-1},y_1,\dots,y_{k-1})) \, \vee (\mathrm{SUCC}(x_0, y_0) \wedge \notag\\
\mathrm{SUCC} (x_1, & \mathit{logn}) \wedge \cdots \wedge \mathrm{SUCC}(x_{k-1},\mathit{logn}) \wedge y_1= 0 \wedge \cdots \wedge y_{k-1}=0)\, \label{arith2}]
\end{align}

It is useful to define an auxiliary predicate $\mathrm{DEF_k}(I)$, where $I$ is a second-order variable of arity and exponent $k$, such that ${\bf A}, \mathit{val} \models \mathrm{DEF_k}(I)$ if $\mathit{val}(I) = B^k$. Please, note that we abuse the notation, writing for instance $\bar{x} = \bar{0}$ instead of $x_0 = 0 \wedge \cdots \wedge x_{k-1} = 0$. Such abuses of notation should nevertheless be clear from the context.   
\begin{align}
\mathrm{DEF}_k(I) & \equiv \exists \bar{x} ( \bar{x} = \bar{0} \wedge I(\bar{x})) \wedge \forall \bar{y} (I(\bar{y}) \rightarrow ((\mathrm{SUCC}(y_0, \mathit{logn}) \wedge \cdots \notag \\
&  \wedge  \mathrm{SUCC}(y_k, \mathit{logn})) \vee \exists \bar{z} (SUCC_k(\bar{y},\bar{z}) \wedge I(\bar{z})))) \label{arith3}
\end{align}

The formula $\mathrm{BIN}_k(X)$, where $X$ is a second-order variable of arity $k+1$ and exponent $k$, expresses that $X$ encodes (as per Definition~\ref{def:repBinaryNumbers}) a binary number between $0$ and $2^{\lceil\log n \rceil^k} - 1$ and can be written as follows. 
\[ \exists I (\mathrm{DEF}_k(I) \wedge \forall \bar{x} (I(\bar{x}) \to ((X(\bar{x},0) \wedge \neg X(\bar{x},1)) \vee (X(\bar{x},1) \wedge \neg X(\bar{x},0))))).\]
However, since $X$ is of exponent $k$, the semantics of $\mathrm{SO}^{\mathit{plog}}$ determines that the number of tuples in any valid valuation of $X$ is always bounded by $\lceil\log n \rceil^k$. Thus $\mathrm{BIN}_k(X)$ can also be expressed by the following equivalent, simpler formula.
\begin{equation}
\mathrm{BIN}_k(X) \equiv \exists I (\mathrm{DEF}_k(I) \wedge \forall \bar{x} (I(\bar{x}) \to (X(\bar{x},0) \vee X(\bar{x},1)))) \label{arith4}
\end{equation}

In the following, $\mathrm{BIN}_k(X, I)$ denotes the sub-formula $\forall \bar{x} (I(\bar{x}) \to (X(\bar{x},0) \vee X(\bar{x},1)))$ of $\mathrm{BIN}_k(X)$. 

The comparison relations $X =_k Y$ and $X <_k Y$ ($X$ is strictly smaller than $Y$) among binary numbers encoded as second-order relations are defined as follows:
\begin{equation}
X =_k Y \equiv \exists I \big(\mathrm{DEF}_k(I) \wedge \mathrm{BIN}_k(X, I) \wedge \mathrm{BIN}_k(Y,I) \wedge {=_k}(X, Y, I) \big), \label{arith5}
\end{equation}

\noindent
where ${=_k}(X, Y, I) \equiv \forall \bar{x} \big( I(\bar{x}) \to \exists z (X(\bar{x},z) \wedge Y(\bar{x}, z))\big)$.
\begin{equation}
X <_k Y \equiv \exists I \big(\mathrm{DEF}_k(I) \wedge \mathrm{BIN}_k(X, I) \wedge \mathrm{BIN}_k(Y,I) \wedge {<_k}(X, Y, I)\big) , \label{arith6} 
\end{equation}

\noindent
where ${<_k}(X, Y, I) \equiv \exists \bar{x} \big( I(\bar{x}) \wedge X(\bar{x}, 0) \wedge Y(\bar{x}, 1) \land \forall \bar{y} \big(I(\bar{y}) \to (\bar{y} \leq_k \bar{x} \vee \exists z (X(\bar{y},z) \wedge Y(\bar{y}, z)))\big)\big)$.

Sometimes we need to determine if the binary number encoded in (the current valuation of) a second-order variable $X$ of arity $k+1$ and exponent $k$ corresponds to the binary representation of an individual $x$ from the domain. The following $\mathrm{BNUM}_k(X,x)$ formula holds whenever that is the case.  
\begin{align}
    \mathrm{BNUM}_k&(X,x) \equiv 
\exists I \big(\mathrm{DEF}_k(I) \wedge \mathrm{BIN}_k(X,I) \wedge \notag \\
&    \forall \bar{y} \big(I(\bar{y}) \to \big((y_0 = 0 \wedge \cdots \wedge y_{k-2} = 0 \wedge (X(\bar{y}, 1) \leftrightarrow BIT(x, y_{k-1}))) \vee \notag \\
&     \hspace*{1.9cm}(\neg (y_0 = 0 \wedge \cdots \wedge y_{k-2} = 0) \wedge X(\bar{y},0))\big)\big)\big) \label{arith7}
\end{align}

We use $\mathrm{BNUM}_k(X,x,I)$ to denote the sub-formula $\forall \bar{y} (I(\bar{y}) \to ((y_0 = 0 \wedge \cdots \wedge y_{k-2} = 0 \wedge (X(\bar{y}, 1) \leftrightarrow BIT(x, y_{k-1}))) \vee (\neg (y_0 = 0 \wedge \cdots \wedge y_{k-2} = 0) \wedge X(\bar{y},0))))$ of $\mathrm{BNUM}_k(X,x)$.

We now proceed to define $\Sigma^{\mathit{plog}}_1$-formulae that describe basic (bounded) arithmetic operations among binary numbers. We start with $\mathrm{BSUM}_k(X,Y,Z)$, where $X$, $Y$ and $Z$ are free-variables of arity $k+1$ and exponent $k$. This formula holds if (the current valuation of) $X$, $Y$ and $Z$ represent binary numbers between $0$ and $2^{\lceil\log n \rceil^k} - 1$, and $X + Y = Z$. The second-order variables $I$ and $W$ in the formula are of arity $k$ and $k+1$, respectively, and both have exponent $k$. We use the traditional carry method, bookkeeping the carried digits in $W$. 
\begin{align}
\mathrm{BSUM}_k&(X,Y,Z) \equiv \notag \\
\exists I  W \big(&\mathrm{DEF}_k(I) \wedge \mathrm{BIN}_k(X,I) \wedge \mathrm{BIN}_k(Y,I) \wedge \mathrm{BIN}_k(Z,I) \wedge \mathrm{BIN}_k(W,I) \wedge \notag \\
&W(\bar{0},0) \wedge ({<_k}(X, Z, I) \vee {=_k}(X, Z, I)) \wedge ({<_k}(Y, Z, I) \vee {=_k}(Y, Z, I)) \wedge \notag \\
&\forall \bar{x} (I(\bar{x}) \to ( (\bar{x} = \bar{0} \wedge \varphi(X, Y, Z)) \vee  \notag\\
&\hspace*{1.9cm} (\exists \bar{y} (SUCC_k(\bar{y},\bar{x}) \wedge \psi(\bar{x}, \bar{y}, W, X, Y)) \wedge  \alpha(\bar{x}, W, X, Y, Z))))\big) \label{arith8}
\end{align}
where $\varphi$ holds if the value of the least significant bit of $Z$ is consistent with the sum of the least significant bits of $X$ and $Y$. Formula $\psi$ holds if the value of the bit in position $\bar{x}$ of $W$ (i.e., the value of the carried bit) is consistent with the sum of the values of the bits in the position preceding $\bar{x}$ of $W,X$ and $Y$. Finally, $\alpha$ holds if the value of the bit in position $\bar{x}$ of $Z$ is consistent with the sum of the corresponding bit values of $W,X$ and $Z$. The actual sub-formulae $\varphi$, $\psi$ and $\alpha$ can be written respectively as follows:
\begin{align*}
\varphi(X, Y, Z) \equiv &
\big( Z(\bar{0},0) \wedge ( (X(\bar{0},0) \wedge Y(\bar{0},0)) \vee (X(\bar{0},1) \wedge Y(\bar{0},1)) ) \big) \vee\\
&\big( Z(\bar{0},1) \wedge ( (X(\bar{0},1) \wedge Y(\bar{0},0)) \vee (X(\bar{0},0) \wedge Y(\bar{0},1)) ) \big)
\end{align*} 
\begin{align*}
\psi(&\bar{x}, \bar{y}, W, X, Y) \equiv \\
&\big( W(\bar{x},0) \wedge ( (W(\bar{y},0)\wedge X(\bar{y},0)\wedge Y(\bar{y},0)) \vee (W(\bar{y},0)\wedge X(\bar{y},0)\wedge Y(\bar{y},1))\vee\\
&\hspace*{1.9cm} (W(\bar{y},0)\wedge X(\bar{y},1)\wedge Y(\bar{y},0)) \vee (W(\bar{y},1)\wedge X(\bar{y},0)\wedge Y(\bar{y},0))) \big) \vee \\
&\big( W(\bar{x},1) \wedge ( (W(\bar{y},1)\wedge X(\bar{y},1)\wedge Y(\bar{y},0)) \vee (W(\bar{y},1)\wedge X(\bar{y},0)\wedge Y(\bar{y},1))\vee\\
&\hspace*{1.9cm} (W(\bar{y},0)\wedge X(\bar{y},1)\wedge Y(\bar{y},1)) \vee (W(\bar{y},1)\wedge X(\bar{y},1)\wedge Y(\bar{y},1))) \big) 
\end{align*} 
\begin{align*}
\alpha(&\bar{x}, W, X, Y, Z) \equiv \\
&\big( Z(\bar{x},0) \wedge ( (W(\bar{x},0)\wedge X(\bar{x},0)\wedge Y(\bar{x},0)) \vee (W(\bar{x},0)\wedge X(\bar{x},1)\wedge Y(\bar{x},1))\vee\\
&\hspace*{1.8cm} (W(\bar{x},1)\wedge X(\bar{x},1)\wedge Y(\bar{x},0)) \vee (W(\bar{x},1)\wedge X(\bar{x},0)\wedge Y(\bar{x},1))) \big) \vee \\
&\big( Z(\bar{x},1) \wedge ( (W(\bar{x},0)\wedge X(\bar{x},0)\wedge Y(\bar{x},1)) \vee (W(\bar{x},0)\wedge X(\bar{x},1)\wedge Y(\bar{x},0))\vee\\
&\hspace*{1.8cm} (W(\bar{x},1)\wedge X(\bar{x},0)\wedge Y(\bar{x},0)) \vee (W(\bar{x},1)\wedge X(\bar{x},1)\wedge Y(\bar{x},1))) \big)
\end{align*}

For the operation of (bounded) multiplication of binary numbers, we define a formula $\mathrm{BMULT}_k(X, Y, Z)$, where $X$, $Y$ and $Z$ are free-variables of arity $k+1$ and exponent $k$. This formula holds if (the current valuations of) $X$, $Y$ and $Z$ represent binary numbers between $0$ and $2^{\lceil\log n \rceil^k} - 1$, and $X \cdot Y = Z$. 

The strategy to express the multiplication consists on keeping track of the (partial) sums of the partial products by means of a relation $R \subset B^k \times B^k \times \{0,1\}$ of size $\lceil \log n \rceil^{2k}$ (recall that $B = \{0, \ldots, \lceil \log n \rceil -1\}$). We take $X$ to be the multiplicand and $Y$ to be the multiplier. Let $\bar{a} \in B^k$ be the $i$-th tuple in the numerical order of $B^k$, let $R|_{\bar{a}}$ denote the restriction of $R$ to those tuples starting with $\bar{a}$, i.e., $R|_{\bar{a}} = \{(\bar{b}, c) \mid (\bar{a}, \bar{b}, c) \in R\}$, and let $\mathit{pred}(\bar{a})$ denote the immediate predecessor of $\bar{a}$ in the numerical order of $B^k$, then the following holds:

\begin{enumerate}[a.]
\item If $\bar{a} = \bar{0}$ and $Y(\bar{a},0)$, then $R|_{\bar{a}}$ encodes the binary number $0$.
\item If $\bar{a} = \bar{0}$ and $Y(\bar{a},1)$, then $R|_{\bar{a}} = X$.
\item If $\bar{a} \neq \bar{0}$ and $Y(\bar{a},0)$, then $R|_{\bar{a}} = R|_{\mathit{pred}(\bar{a})}$.
\item If $\bar{a} \neq \bar{0}$ and $Y(\bar{a},1)$, then (the binary number encoded by) $R|_{\bar{a}}$ results from adding $R|_{\mathit{pred}(\bar{a})}$ to the $(i-1)$-bits arithmetic left-shift of $X$.
\end{enumerate}

$\mathrm{BMULT}_k(X, Y, Z)$ holds if $Z = R|_{(a_0, \ldots, a_{k-1})}$ for $a_0 = \cdots = a_{k-1} = \lceil \log n \rceil -1$. 
Following this strategy, we can write $\mathrm{BMULT}_k(X, Y, Z)$ as follows. 
\begin{align}
\mathrm{BMULT}_k & (X,Y,Z) \equiv \notag \\
\exists I I' & R S W \big(\mathrm{DEF}_k(I) \wedge \mathrm{BIN}_k(X,I) \wedge \mathrm{BIN}_k(Y,I) \wedge \mathrm{BIN}_k(Z,I) \wedge \notag \\
&\mathrm{DEF}_{2k}(I') \wedge \mathrm{BIN}_{2k}(R,I') \wedge \mathrm{BIN}_{2k}(S,I') \wedge \mathrm{BIN}_{2k}(W,I') \wedge \notag \\
&\mathrm{SHIFT}(S, X, I) \wedge \notag \\
&\forall \bar{x} (I(\bar{x}) \to ( (\bar{x} = \bar{0} \wedge Y(\bar{x},0) \wedge \varphi_a(R, \bar{x})) \vee \notag \\
&\hspace*{1.5cm} (\bar{x} = \bar{0} \wedge Y(\bar{x},1) \wedge \varphi_b(R, \bar{x}, X)) \vee \notag \\
&\hspace*{1.5cm} (Y(\bar{x},0) \wedge \exists \bar{y} (SUCC_k(\bar{y},\bar{x}) \wedge \varphi_c(R, \bar{x}, \bar{y}))) \vee \notag \\
&\hspace*{1.5cm} (Y(\bar{x},1) \wedge \exists \bar{y} (SUCC_k(\bar{y},\bar{x}) \wedge \varphi_d(R, S, W, \bar{x}, \bar{y})))))\big) \label{arith9}
\end{align}
Here the variable $I$ has arity $k$ and exponent $k$. $I'$ is of arity $2k$ and exponent $2k$. The remaining second-order variables $R$, $S$ and $W$ are of arity $2k + 1$ and exponent $2k$.
The sub-formula $\varphi_a(R, \bar{x}) \equiv \forall \bar{y} (I(\bar{y}) \to R(\bar{x}, \bar{y}, 0))$  expresses that $R|_{\bar{x}}$ encodes the binary number $0$, the sub-formula $\varphi_b(R, \bar{x}, X) \equiv \forall \bar{y} (I(\bar{y}) \to \exists z(R(\bar{x}, \bar{y}, z) \wedge X(\bar{y}, z)))$ expresses that $R|_{\bar{x}} = X$, the sub-formula $\varphi_c(R, \bar{x}, \bar{y}) \equiv \forall \bar{w} (I(\bar{w}) \to \exists z(R(\bar{x}, \bar{w}, z) \wedge R(\bar{y}, \bar{w}, z)))$ expresses that $R|_{\bar{x}} = R|_{\bar{y}}$, and the sub-formula $\mathrm{SHIFT}(S, X, I)$ expresses that if $\bar{a} \in B^k$ is the $i$-th tuple in the numerical order of $B^k$, then $S|_{\bar{a}}$ is the $(i-1)$-bits arithmetic left-shift of $X$, i.e., $S|_{\bar{a}}$ is $X$ multiplied by $2^{i-1}$ in binary. 
\begin{align}
\mathrm{SHIFT}(S, X&, I) \equiv \exists x \bar{y} \big(\mathrm{SUCC}(x, logn) \wedge S(\bar{y}, \bar{x}, 0) \big) \wedge \notag \\
\forall \bar{x} \big(I(\bar{x}) \to &( (\bar{x} = \bar{0} \wedge \varphi_b(S,\bar{x}, X)) \vee \notag \\
&\hspace*{0.2cm} \exists \bar{y} (\mathrm{SUCC}_k(\bar{y},\bar{x}) \wedge \notag \\
&\hspace*{0.8cm} \forall \bar{z} (I(\bar{z}) \to ((\bar{z} = \bar{0} \wedge S(\bar{x}, \bar{z}, 0)) \vee \notag \\
&\hspace*{1.7cm} \exists \bar{z}' b (\mathrm{SUCC}_k(\bar{z}',\bar{z}) \wedge S(\bar{y}, \bar{z}', b) \wedge S(\bar{x}, \bar{z}, b))))))\big)   \label{arith10}  
\end{align}
Finally, the sub-formula $\varphi_d(R, S, W, \bar{x}, \bar{y})$ expresses that $R|_{\bar{x}}$ results from adding $R|_{\bar{y}}$ to $S|_{\bar{x}}$. The carried digits of this sum are kept in $W|_{\bar{x}}$. Given the formula $\mathrm{BSUM}_k$ described earlier, it is a straightforward task to write $\varphi_d$. We omit further details.  

The operations of division and modulo are expressed by $\mathrm{BDIV}_k(X, Y, Z, M)$, where $X$, $Y$, $Z$ and $M$ are free-variables of arity $k+1$ and exponent $k$. This formula holds if $Z$ is the quotient and $M$ the modulo (remainder) of the euclidean division of $X$ by $Y$, i.e., if  $Y\cdot Z + M = X$. 
 \begin{align}
    \mathrm{BDIV}_k(X, Y, Z, M)\equiv \notag \\
\exists I I' A R S W W' \big(&\mathrm{DEF}_k(I) \wedge \mathrm{BIN}_k(X,I) \wedge \mathrm{BIN}_k(Y,I) \wedge \mathrm{BIN}_k(Z,I) \wedge \notag \\
    &\mathrm{BIN}_k(M,I) \wedge \mathrm{BIN}_k(A,I) \wedge \neg \mathrm{BNUM}_k(Y,0,I) \wedge \notag \\
    &{<_k}(M, Y, I) \wedge \mathrm{BMULT}_k(Z, Y, A, I, I', R, S, W) \wedge \notag \\
    &\mathrm{BSUM}_k(A,M,X, I, W')\big). \label{arith11}
\end{align}
where $\mathrm{BMULT}_k(Z, Y, A, I, I', R, S, W)$ denotes the formula obtained from the formula $\mathrm{BMULT}_k(X, Y, Z)$ in~\ref{arith9} by deleting the second-order quantifiers (so that $I$, $I'$, $R$, $S$ and $W$ become free-variables) and by renaming $X$ and $Z$ as $Z$ and $A$, respectively. Likewise, $\mathrm{BSUM}_k(A,M,X, I, W')$ denotes the formula obtained from $\mathrm{BSUM}_k(X, Y, Z)$ in~\ref{arith8} by deleting the second-order quantifiers (so that $I$ and $W$ become free-variables) and by renaming $X$, $Y$, $Z$ and $W$ as $A$, $M$, $X$ and $W'$, respectively.

\section{The Poly-logarithmic Time Hierarchy}\label{sec:plh}

The sequential access that Turing machines have to their tapes makes it impossible to compute anything in sub-linear time. Therefore, logarithmic time complexity classes are usually studied using models of computation that have random access to their input. As this also applies to the poly-logarithmic complexity classes studied in this paper, we adopt a Turing machine model that has a \emph{random access} read-only input, similar to the log-time Turing machine in~\cite{barrington:jcss1990}.

A \emph{random-access Turing machine} is a multi-tape Turing machine with (1) a read-only (random access) \emph{input} of length $n+1$, (2) a fixed number of read-write \emph{working tapes}, and (3) a read-write input \emph{address-tape} of length $\lceil \log n \rceil$.

Every cell of the input as well as every cell of the address-tape contains either $0$ or $1$ with the only exception of the ($n+1$)st cell of the input, which is assumed to contain the endmark $\triangleleft$. In each step the binary number in the address-tape either defines the cell of the input that is read or if this number exceeds $n$, then the ($n+1$)st cell containing $\triangleleft$ is read.   

\begin{example}\label{ex:machine}
Let polylogCNFSAT be the restriction of the class CNFSAT (aka CNF) of satisfiable propositional formulae in conjunctive normal form to $c \leq \lceil \log n \rceil^k$ clauses, where $n$ is the length of the formula. Note that the formulae in polylogCNFSAT tend to have few clauses and many literals. We define a random-access Turing machine $M$ which decides polylogCNFSAT. The alphabet of $M$ is $\{0,1,\#,+,-\}$. The input formula is encoded in the input tape as a list of $c \leq \lceil \log n \rceil^k$ indices (binary numbers of length $\lceil \log n \rceil$), followed by $c$ clauses. For every $1 \leq i \leq c$, the $i$-th index points to the first position in the $i$-th clause. Clauses start with $\#$ and are followed by a list of literals. Positive literals start with a $+$, negative with a $-$. The $+$ or $-$ symbol of a literal is followed by the ID of the variable in binary. $M$ proceeds as follows: (1) Using binary search with the aid of the ``out of range'' response $\triangleleft$, compute $n$ and $\lceil \log n \rceil$. (2) Copy the indices to a working tape, counting the number of indices (clauses) $c$. (3) Non-deterministically guess $c$ input addresses $a_1, \ldots, a_c$, i.e., guess $c$ binary numbers of length $\lceil \log n \rceil$. (4) Using $c$ $1$-bit flags, check that each $a_1, \ldots, a_c$ address falls in the range of a different clause. (5) Check that each $a_1, \ldots, a_c$ address points to an input symbol $+$ or $-$. (6) Copy the literals pointed by $a_1, \ldots, a_c$ to a working tape, checking that there are \emph{no} complementary literals. (7) Accept if all checks hold.
\end{example}

Let $L$ be a language accepted by a random-access Turing machine $M$. Assume that for some function $f$ on the natural numbers, $M$ makes at most $O(f(n))$ steps before accepting an input of length $n$. If $M$ is deterministic, then we write $L \in \mathrm{DTIME}[f(n)]$. If $M$ is non-deterministic, then we write $L \in \mathrm{NTIME}[f(n)]$. We define the classes of deterministic and non-deterministic poly-logarithmic time computable problems as follows:
\[ \polylog = \bigcup_{k \in \mathbb{N}} \mathrm{DTIME}[\log^k n] \qquad \, \npolylog = \bigcup_{k \in \mathbb{N}} \mathrm{NTIME}[\log^k n] \]
The non-deterministic random-access Turing machine in Example~\ref{ex:machine} clearly works in polylog-time. Therefore, polylogCNFSAT $\in \npolylog$.

In order to relate our logic $\mathrm{SO}^{\mathit{plog}}$ to these Turing complexity classes we adhere to the usual conventions concerning a binary encoding of finite structures~\cite{Immerman99}. Let $\sigma = \{R^{r_1}_1, \ldots, R^{r_p}_p, c_1, \ldots, c_q\}$ be a vocabulary, and let ${\bf A}$ with $A = \{0, 1, \ldots, n-1\}$ be an ordered structure of vocabulary $\sigma$. Each relation $R_i^{\bf A} \subseteq A^{r_i}$ of $\bf A$ is encoded as a binary string $\mathrm{bin}(R^{\bf A}_i)$ of length $n^{r_i}$ where $1$ in a given position indicates that the corresponding tuple in the lexicographical ordering is in $R_i^{\textbf{A}}$.
Likewise, each constant number $c^{\bf A}_j$ is encoded as a binary string $\mathrm{bin}(c^{\bf A}_j)$ of length $\lceil \log n \rceil$. The encoding of the whole structure $\mathrm{bin}(\textbf{A})$ is simply the concatenation of the binary strings encodings its relations and constants: 
\[\mathrm{bin}(\textbf{A}) = \mathrm{bin}(R_1^{\textbf{A}}) \cdots \mathrm{bin}(R_p^{\textbf{A}}) \cdot \mathrm{bin}(c^{\bf A}_1) \cdots \mathrm{bin}(c^{\bf A}_q).\] 

The length $\hat{n} = |\mathrm{bin}(\textbf{A})|$ of this string is $n^{r_1}+\cdots+n^{r_p} + q \lceil \log n \rceil$, where $n = |A|$ denotes the size of the input structure ${\bf A}$. Note that $\log \hat{n} \in O(\lceil \log n \rceil)$, so $\mathrm{NTIME}[\log^k \hat{n}] = \mathrm{NTIME}[\log^k n]$ (analogously for $\mathrm{DTIME}$). Therefore, we will consider random-access Turing machines, where the input is the encoding $\mathrm{bin}(\textbf{A})$ of the structure \textbf{A} followed by the endmark $\triangleleft$.

In this work we also consider alternating Turing machines. An alternating Turing machine comes with a set of states $Q$ that is partitioned into subsets $Q_\exists$ and $Q_\forall$ of so-called existential and universal states. Then a configuration $c$ is accepting iff
\begin{itemize}

\item $c$ is in a final accepting state,

\item $c$ is in an existential state and there exists a next accepting configuration, or

\item $c$ is in a universal state, there exists a next configuration and all next configurations are accepting.

\end{itemize}

In analogy to our definition above we can define a \emph{random-access alternating Turing machine}. The languages accepted by such a machine $M$, which starts in an existential state and makes at most $O(f(n))$ steps before accepting an input of length $n$ with at most $m$ alternations between existential and universal states, define the complexity class $\mathrm{ATIME}[f(n),m]$. Analogously, we define the complexity class $\mathrm{ATIME}^{op}[f(n),m]$ comprising languages that are accepted by a random-access alternating Turing machine that starts in a universal state and makes at most $O(f(n))$ steps before accepting an input of length $n$ with at most $m$ alternations between universal and existential states. With this we define
\[ \tilde{\Sigma}_m^{\mathit{plog}} = \bigcup_{k \in \mathbb{N}} \mathrm{ATIME}[\log^k n,m] \quad \text{and} \quad \tilde{\Pi}_m^{\mathit{plog}} = \bigcup_{k \in \mathbb{N}} \mathrm{ATIME}^{op}[\log^k n,m] . \]

The poly-logarithmic time hierarchy is then defined as $\mathrm{PLH} = \bigcup_{m \ge 1} \tilde{\Sigma}_m^{\mathit{plog}}$. Note that $\tilde{\Sigma}_1^{\mathit{plog}} = \npolylog$ holds. 

\begin{remark}

Note that a simulation of a $\npolylog$ Turing machine $M$ by a deterministic machine $N$ requires checking all computations in the tree of computations of $M$. As $M$ works in time $({\log n})^{O(1)}$, $N$ requires time $2^{{\log n}^{O(1)}}$. This implies $\npolylog \subseteq \mathrm{DTIME}(2^{{\log n}^{O(1)}})$, which is the complexity class called quasipolynomial time of the fastest known algorithm for graph isomorphism \cite{babai:stoc2016}, which further equals the class  
$\mathrm{DTIME}({n^{{\log n}^{O(1)}}})$\footnote{This relationship appears quite natural in view of the well known relationship $\mathrm{NP} = \mathrm{NTIME}(n^{O(1)}) \subseteq \mathrm{DTIME}(2^{{n}^{O(1)}}) = \mathrm{EXPTIME}$.}.

\end{remark}

\section{Correspondence Between the Quantifier Prefix Classes of \texorpdfstring{$\mathrm{SO}^{\mathit{plog}}$}{TEXT} and the Levels of the Polylog-Time Hierarchy} \label{sec:main}

We say that a logic $\mathcal{L}$ captures the complexity class $\mathcal{K}$ iff the following holds:
\begin{itemize} 

\item For every $\mathcal{L}$-sentence $\varphi$ the language $\{\mathrm{bin}({\bf A}) \mid {\bf A} \models \varphi \}$ is in $\mathcal{K}$, and

\item For every property $\mathcal{P}$ of (binary encodings of) structures that can be decided with complexity in $\mathcal{K}$, there is a sentence $\varphi_{\mathcal{P}}$ of $\mathcal{L}$ such that ${\bf A} \models \varphi_{\mathcal{P}}$ iff $\bf A$ has the property $\mathcal{P}$. 

\end{itemize} 

We now present our main result which states that the existential fragment of $\mathrm{SO}^{\mathit{plog}}$ captures $\npolylog$.

\begin{theorem}\label{pedsoplog}

Over ordered structures with successor relation, $\mathrm{BIT}$ and constants for $\log n$, the minimum, second and maximum elements, $\Sigma^{\mathit{plog}}_1$ captures $\npolylog$.

\end{theorem}

\begin{proof}
\ \textbf{Part a.} We first show $\Sigma^{\mathit{plog}}_1 \subseteq NPolyLogTime$, i.e. a non-deterministic random access Turing Machine \textbf{M} can evaluate every sentence $\phi$ in $\Sigma^{\mathit{plog}}_1$ in poly-logarithmic time. 

Let $\phi = \exists X_1^{r_1,\log^{k_1}} \dots \exists X_m^{r_m,\log^{k_m}} \varphi$, where $\varphi$ is a first-order formula with the restrictions given in the definition of $\mathrm{SO}^{\mathit{plog}}$. Given a $\sigma$-structure \textbf{A} with $|dom(\textbf{A})|=n$, \textbf{M} first guesses values for $X_1^{r_1,\log^{k_1}},\dots, X_m^{r_m,\log^k_m}$ and then checks if $\varphi$ holds. As $val(X_i^{r_i,\log^{k_i}})$ is a relation of arity $r_i$ with at most $\log^{k_i} n$ tuples, \textbf{M} has to guess $r_i*{(\log(n))}^{k_i}$ values in $A$, each encoded in $\left \lceil log(n)\right \rceil$ bits. Thus, the machine has to generate $E = \sum\limits_{i=1}^{m}\left ( r_i*{(\log(n))}^{k_i+1}\right )$ bits in total. As $E \in O(\lceil \log n \rceil^{k_{\max} + 1})$, the generation of the values $val(X_i^{r_i,\log^{k_i}})$ requires time in $O(\lceil \log n \rceil^{k^\prime})$ for some $k^\prime$.

The fact that \textbf{M} can check the validity of $\varphi$ in poly-logarithmic time, can be shown by structural induction on the formulae.   
\begin{enumerate}[i.]
\item If $\varphi$ is an existential first-order formula, then \textbf{M} can clearly check $\varphi$ in poly-logarithmic time. Note that if $\varphi$ is an atomic formulae, it only takes time $O(\log n)$ to \textbf{M} to decide whether $\varphi$ holds or not. For reference see proof of~\cite[Theorem 5.30]{Immerman99} among others.
\item If $\varphi$ is of the form $X^{r,\log^k}(t_1, \ldots, t_r)$ or $\neg X^{r,\log^k}(t_1, \ldots, t_r)$, then \textbf{M} can simply check whether the tuple $(t_1, \ldots, t_r)$ belongs to the relation assigned to $X^{r,\log^k}$. Since this relation is of maximum size $\lceil\log n\rceil^k$, then this takes poly-logarithmic time.
\item If $\varphi$ is of the form $\psi_1 \vee \psi_2$ (or $\varphi = \psi_1 \wedge \psi_2$), then \textbf{M} has to check if $\psi_1$ or $\psi_2$ (or both, respectively) holds, which requires at most the time for checking both $\psi_1$ and $\psi_2$. Thus, by the inductive hypothesis the checking of $\varphi$ can be done in poly-logarithmic time.
\item If $\varphi$ is of the form $\forall \bar{x} (X^{r,\log^k}(\bar{x}) \rightarrow \psi)$, then $\textbf{M}$ has already guessed a value for $X^{r,\log^{k}}$, it remains to check whether $\{ \bar{a} / \bar{x} \} . \psi$ for every element $\bar{a}$ in this relation. Since there are maximum $\lceil \log^k n \rceil$ tuples in the valuation of $X^{r,\log^{k}}$, by the inductive hypothesis the whole process takes poly-logarithmic time.
\item If $\varphi$ is of the form $\exists x \psi$, \textbf{M} first guesses $x$, for which at most $\lceil \log n \rceil$ steps are required, and then checks that $\psi$ holds, which by the inductive hypothesis can be done in poly-logarithmic time.  
\end{enumerate}

\noindent
\textbf{Part b.} Next we show $\npolylog  \subseteq \Sigma^{\mathit{plog}}_1$. For this let \textbf{M} be a non-deterministic random access Turing Machine that accepts a $\sigma$-structure \textbf{A} in $O(\log^k n)$ steps, where $|\mathit{dom}({\bf A})| = n$. We assume a set of states $Q = \{ q_0,\dots,q_f \}$, where $q_0$ is the initial state, $q_f$ is the only final state. In the initial state, the tape heads are in the left-most position, the working tape is empty and the index-tape is filled with zeros.

As \textbf{M} runs in time $\left\lceil \log n\right\rceil^k$, it visits at most $\left\lceil \log n\right\rceil^k$ cells in the working tape. Thus, we can model positions on the working tape and time by $k$-tuples $\bar{p}$ and $\bar{t}$, respectively. Analogously, the length of the index tape is bound by $\left\lceil \log n\right\rceil^{k'} $, so we can model the positions in the index tape by $k'$-tuples $\bar{d}$. We use auxiliary relations $I$ and $I^\prime$ to capture $k$-tuples and $k^\prime$-tuples, respectively, over $\{ 0 ,\dots, \lceil \log n \rceil \}$.  We define those relations using $\mathrm{DEF_k}(I)$ and $\mathrm{DEF_{k'}}(I')$ in the same way as in (\ref{arith3}).
As \textbf{M} works non-deterministically, it makes a choice in every step. Without loss of generality we can assume that the choices are always binary, which we capture by a relation $C$ of arity $k + k^\prime + 1$; $C(\bar{t},\bar{d},c)$ expresses that at time $\bar{t}$ the position $\bar{d}$ in the index-tape has the value $c \in \{0,1\}$, which denotes the two choices.

In order to construct a sentence in $\Sigma^{\mathit{plog}}_1$ that is satisfied by the structure \textbf{A} iff the input $\mathrm{bin}(\textbf{A})$ is accepted by \textbf{M} we first describe logically the operation of the random access Turing machine \textbf{M}, then express the acceptance of $\mathrm{bin}(\textbf{A})$ for at least one computation path.

We use predicates $T_0, T_1, T_2$, where $T_i(\bar{t},\bar{p})$ indicates that at time $\bar{t}$ the working tape at position $\bar{p}$ contains $i$ for $i \in \{ 0,1 \}$ and the blank symbol for $i=2$, respectively. The following formulae express that the working tape is initially empty, and at any time a cell can only contain one of the three possible symbols:
\begin{align}
& \forall \bar{p} \; I(\bar{p}) \rightarrow T_2(\bar{0},\bar{p}) \notag \\
& \forall \bar{t} I(\bar{t}) \rightarrow \forall \bar{p} \; I(\bar{p}) \rightarrow ( T_0(\bar{t},\bar{p}) \rightarrow \neg T_1(\bar{t},\bar{p}) \wedge \neg T_2(\bar{t},\bar{p}) ) \notag \\
& \forall \bar{t} I(\bar{t}) \rightarrow \forall \bar{p} \; I(\bar{p}) \rightarrow ( T_1(\bar{t},\bar{p}) \rightarrow \neg T_0(\bar{t},\bar{p}) \wedge \neg T_2(\bar{t},\bar{p}) ) \notag \\
& \forall \bar{t} I(\bar{t}) \rightarrow \forall \bar{p} \; I(\bar{p}) \rightarrow ( T_2(\bar{t},\bar{p}) \rightarrow \neg T_0(\bar{t},\bar{p}) \wedge \neg T_1(\bar{t},\bar{p}) ) \label{eq-2}
\end{align}

Then we use a predicate $H$ with $H(\bar{t},\bar{p})$ expressing that at time $\bar{t}$ the head of the working tape is in position $\bar{p}$. This gives rise to the following formulae:
\begin{gather}
H(\bar{0},\bar{0}) \qquad\qquad\qquad \forall \bar{t} I(\bar{t}) \rightarrow \exists \bar{p} ( I(\bar{p}) \wedge H(\bar{t},\bar{p}) ) \notag \\
\forall \bar{t} I(\bar{t}) \rightarrow \forall \bar{p} \; I(\bar{p}) \rightarrow ( H(\bar{t},\bar{p}) \rightarrow \forall \bar{p}^\prime ( I(\bar{p}^\prime) \rightarrow H(\bar{t},\bar{p}^\prime) \rightarrow \bar{p} = \bar{p}^\prime ) ) \label{eq-3}
\end{gather}

Predicates $S_i$ for $i = 1 ,\dots, f$ are used to express that at time $\bar{t}$ the machine \textbf{M} is in the state $q_i \in Q$, which using (\ref{arith1}) gives rise to the formulae
\begin{gather}
S_0(\bar{0}) \wedge \forall \bar{t}\big( I(\bar{t}) \rightarrow ( \bar{t} \neq \bar{0} \rightarrow \bigvee_{0 \le i \le f} S_i(\bar{t}) )\big) \wedge 
\exists \bar{t}_f \Big( \forall \bar{t} \big( I(\bar{t}) \rightarrow \bar{t} \le_k \bar{t}_f \big)  \wedge S_f(\bar{t}_f)\Big) \notag \\
\bigwedge_{0 \le i \le f} \; \forall \bar{t} \Big(I(\bar{t}) \rightarrow \big(S_i(\bar{t}) \rightarrow \bigwedge_{0 \le j \le f, j \neq i} \neg S_j(\bar{t})\big)\Big) \label{eq-4}
\end{gather}

The following formulae exploiting (\ref{arith2}) describe the behaviour of \textbf{M} moving in every step its working-tape head either to the right, to the left or not at all (which actually depends on the value for $c$ in $C(\bar{t},\bar{d},c)$):
\begin{align}
& \forall \bar{t} I(\bar{t}) \wedge \bar{t} \neq \bar{0} \wedge H(\bar{t},\bar{0}) \rightarrow
\exists \bar{t}^\prime, \bar{d}^\prime \; ( \mathrm{SUCC}_k(\bar{t}^\prime,\bar{t}) \wedge \notag \\
& \hspace*{3cm} \mathrm{SUCC}_k(\bar{0},\bar{d}^\prime) \wedge ( H(\bar{t}^\prime,\bar{0}) \vee H(\bar{t}^\prime,\bar{d}^\prime) ) ) \notag \\
& \forall \bar{t}, \bar{d} I(\bar{t}) \wedge \bar{t} \neq \bar{0} \wedge I^\prime(\bar{d}) \wedge H(\bar{t},\bar{d}) \rightarrow ( \bar{d} \neq \bar{0} \wedge \bar{d} \neq \mathit{last} \rightarrow \notag \\
& \hspace*{1.5cm} \exists \bar{t}^\prime, \bar{d}_1 , \bar{d}_2 \; ( \mathrm{SUCC}_k(\bar{t}^\prime,\bar{t}) \wedge \mathrm{SUCC}_k(\bar{d}_1,\bar{d}) \wedge \mathrm{SUCC}_k(\bar{d},\bar{d}_2) \wedge \notag \\
& \hspace*{3cm}  ( H(\bar{t}^\prime,\bar{d}_1) \vee H(\bar{t}^\prime,\bar{d}) \vee H(\bar{t}^\prime,\bar{d}_2) ) ) \notag \\
& \forall \bar{t} I(\bar{t}) \wedge \bar{t} \neq \bar{0} \wedge H(\bar{t},\mathit{last}) \rightarrow
\exists \bar{t}^\prime, \bar{d}^\prime \; ( \mathrm{SUCC}_k(\bar{t}^\prime,\bar{t}) \wedge \notag \\
& \hspace*{3cm} \mathrm{SUCC}_k(\bar{d}^\prime,\mathit{last}) \wedge ( H(\bar{t}^\prime,\mathit{last}) \vee H(\bar{t}^\prime,\bar{d}^\prime) ) ) \label{eq-5}
\end{align}
    
Furthermore, we use predicates $L_i$ ($i \in \{ 0, 1, 2 \}$) to describe that \textbf{M} reads at time $\bar{t}$ the value $i$ (for $i \in \{ 0, 1 \}$) or $\triangleleft$ for $i=2$, respectively. As exactly one of these values is read, we obtain the following formulae:
\begin{alignat}{2}
& \forall \bar{t} \; I(\bar{t}) \rightarrow ( L_0(\bar{t}) \vee L_1(\bar{t}) \vee L_2(\bar{t}) ) \quad &
& \forall \bar{t} \; I(\bar{t}) \rightarrow ( L_0(\bar{t}) \rightarrow \neg L_1(\bar{t}) \wedge \neg L_2(\bar{t}) ) \notag \\
& \forall \bar{t} \; I(\bar{t}) \rightarrow ( L_1(\bar{t}) \rightarrow \neg L_0(\bar{t}) \wedge \neg L_2(\bar{t}) ) \quad &
& \forall \bar{t} \; I(\bar{t}) \rightarrow ( L_2(\bar{t}) \rightarrow \neg L_0(\bar{t}) \wedge \neg L_1(\bar{t}) ) \label{eq-6}
\end{alignat}

The conjunction of the formulae in (\ref{eq-2})-(\ref{eq-6}) with all second-order variables existentially quantified merely describes the operation of the Turing machine \textbf{M}. If \textbf{M} accepts the input $\mathrm{bin}(\textbf{A})$ for at least one computation path, i.e. for one sequence of the choices, we can assume without loss of generality that if at time $\bar{t}$ with $\bar{d}$ on the index-tape the bit $c$ indicating the choice equals the value read from the input, then this will lead to acceptance. Therefore, in order to complete the construction of the required formula in $\Sigma^{\mathit{plog}}_1$ we need to express this condition in our logic.

The bit \textbf{M} reads from the input corresponds to the binary encoding of the relations and constants in the structure $\mathbf{A}$. In order to detect, which tuple or which constant is actually read, we require several auxiliary predicates. We use predicates $M_i$ ($i=0, \dots, k^\prime$) to represent the numbers $n^i$, which leads to the formulae
\begin{gather}
\mathrm{BNUM}_{k^\prime}(M_0,1,I^\prime),\;\mathrm{BNUM}_{k^\prime}(M_1,\max,I^\prime) \;\text{and }
 \mathrm{BMULT}_{k^\prime}(M_1,M_{i-1},M_i) \;\text{for}\; i \ge 2 \label{eq-7}
\end{gather}

For this we exploit the definitions in (\ref{arith7}) and (\ref{arith9}). The latter formula is not in QNF. Nevertheless, we have already shown at the end of Section~\ref{sec:examples} how to turn such a formula into a formula in QNF in $\Sigma^{\mathit{plog}}_1$. The same applies to several of the following formulae.

Next we use relations $P_i$ ($i=0,\dots,p+1$) representing the position in $\mathrm{bin}(\textbf{A})$, where the encoding of $R_{i+1}^{\mathbf{A}}$ for the relation $R_i$ starts (for $0 \le i \le p-1$), the encoding of where the constants $c_j$ ($j=1,\dots,q$) starts (for $i=p$), and finally representing the length of $\mathrm{bin}(\textbf{A})$ (for $i=p+1$). As each constant requires $\lceil \log n \rceil$ bits we further use auxiliary relations $N_i$ (for $i=1,\dots,q$) to represent $i \cdot \lceil \log n \rceil$, that is required to detect which constant is read. This leads to the following formulae (exploiting (\ref{arith7}) and (\ref{arith8})):
\begin{align}
& \mathrm{BNUM}_{k^\prime}(P_0,0,I^\prime) \; , \; \bigwedge_{1 \le i \le p} \mathrm{BSUM}_{k^\prime}(P_{i-1},M_{r_i},P_i) \; \text{and} \; \mathrm{BSUM}_{k^\prime}(P_p,N_q,P_{p+1}) \notag \\
& \mathrm{BNUM}_{k^\prime}(N_1,\mathit{logn},I^\prime) \quad \text{and} \; \bigwedge_{1 \le i \le q} \mathrm{BSUM}_{k^\prime}(N_{i-1},N_1,N_i) \label{eq-8}
\end{align}

Finally, we can express the acceptance condition linking the relation $C$ to the input $\mathrm{bin}(\textbf{A})$. In order to ease the representation we use for fixed $\bar{t}$ the shortcut $C_{\bar{t}}$ with $C_{\bar{t}}(\bar{d},c) \leftrightarrow C(\bar{t},\bar{d},c)$. Likewise we use shortcuts with subscript $\bar{t}$ for additional auxiliary predicates $D_i$ ($i=0,\dots,p$), $Q_i$, $Q_i^\prime$, $Q_i^{\prime\prime}$ and $Q_i^{\prime\prime\prime}$ ($i=1,\dots,r_{max}$) which we need for arithmetic operations on the length of $\mathrm{bin}(\textbf{A})$, which is represented by $P_{p+1}$. We also use $\le_{k^\prime}(X,Y,I)$ as shortcut for $<_{k^\prime}(X,Y,I) \vee =_{k^\prime}(X,Y,I)$ defined in~(\ref{arith5}) and~(\ref{arith6}).

For fixed $\bar{t}$ with $I(\bar{t})$ the relation $C_{\bar{t}}$ represents a position in the bitstring $\mathrm{bin}(\textbf{A})$, which is either at the end, within the substring encoding the constants $c_j^{\mathbf{A}}$, or within the substring encoding the relation $R_i^{\mathbf{A}}$. The following three formulae (using (\ref{arith6}), (\ref{arith7}), (\ref{arith8}), and (\ref{arith11})) with fixed $\bar{t}$ correspond to these cases:
\begin{align}
    & <_{k^\prime}(P_{p+1},C_{\bar{t}},I^\prime) \rightarrow L_2(\bar{t}) \notag \\
    & <_{k^\prime}(P_p,C_{\bar{t}},I^\prime) \wedge \le_{k^\prime}(C_{\bar{t}},P_{p+1},I^\prime) \wedge \mathrm{BSUM}_{k^\prime}(P_p,D_{0,\bar{t}},C_{\bar{t}}) \wedge \notag \\
      & \mathrm{BDIV}_{k^\prime}(D_{0,\bar{t}}, N_1, Q_{1,\bar{t}}, Q_{1,\bar{t}}^\prime)\to\exists xy \big( \mathrm{BNUM}_{k^\prime}(Q_{1,\bar{t}},x) \wedge \mathrm{BNUM}_{k^\prime}(Q_{1,\bar{t}}^\prime,y)
      \wedge \notag\\
      & \hspace*{5.6cm}\mathrm{BIT}(c_x,y) {\leftrightarrow} L_1(\bar{t})\big)
    \notag\\
    & \bigwedge_{1 \le i \le p} <_{k^\prime}(P_{i-1},C_{\bar{t}},I^\prime) \wedge \le_{k^\prime}(C_{\bar{t}},P_i,I^\prime) \wedge \mathrm{BSUM}_{k^\prime}(P_{i-1},D_{i,\bar{t}},C_{\bar{t}}) \rightarrow \notag \\
    & \quad \exists \bar{x} \Bigg( \bigwedge_{1 \le j \le r_i} \Big( \mathrm{BNUM}_{k^\prime}(Q_{j,\bar{t}}^{\prime\prime\prime},x_j)
    \wedge \mathrm{BDIV}_{k^\prime}(D_{i,\bar{t}}, M_j, Q_{j,\bar{t}}, Q_{j,\bar{t}}^\prime)\notag\\
    & \quad \quad \quad\quad \quad \quad\quad \quad \quad\quad \quad \quad\wedge \mathrm{BDIV}_{k^\prime}(Q_{j,\bar{t}}, M_1, Q_{j,\bar{t}}^{\prime\prime}, Q_{j,\bar{t}}^{\prime\prime\prime}) \Big)\wedge
    \notag\\
    & \quad\quad\quad\quad \Big(\big(
    L_1(\bar{t}) ) \to  R_i(\bar{x})
    \big)\vee\big(
    L_0(\bar{t}) ) \to  \neg R_i(\bar{x})
    \big) \Big) \Bigg)
    \label{eq-9}
\end{align}
Note that in the second case $Q_{1,\bar{t}}$ represents an index $j \in \{ 1,\dots,q \}$ and $Q_{1,\bar{t}}^\prime$ represents the read bit of the constant $c_j^{\mathbf{A}}$ in $\mathrm{bin}(\textbf{A})$. In the third case $D_{i,\bar{t}}$ represents the read position $d$ in the encoding on relation $R_i^{\mathbf{A}}$, which represents a particular tuple, for which we use $Q_{j,\bar{t}}^{\prime\prime\prime}$ to determine every value of the tuple and depending of the read in $L_i(\bar{t})$ check if that particular tuple is in the relation or not.

Finally, the sentence $\Psi$ describing acceptance by \textbf{M} results from building the conjunction of the formulae in (\ref{eq-2})-(\ref{eq-9}), expanding the macros as shown in Section \ref{sec:examples}, which brings in additional second-order variables, and existentially quantifying all second-order variables. Due to our construction we have $\mathbf{A} \models \Psi$ iff $\mathbf{A}$ is accepted by \textbf{M}.\end{proof}

Likewise, we can show that the universal fragment of $\mathrm{SO}^{\mathit{plog}}$ captures the (first) level $\tilde{\Pi}^{\mathit{plog}}_1$ of the polylog-time hierarchy.

\begin{theorem}\label{pedsoplog2}

Over ordered structures with successor relation, $\mathrm{BIT}$ and constants for $\log n$, the minimum, second and maximum elements, $\Pi^{\mathit{plog}}_1$ captures $\tilde{\Pi}^{\mathit{plog}}_1$.

\end{theorem}

\begin{proof}(sketch)
In order to show $\Pi^{\mathit{plog}}_1 \subseteq \tilde{\Pi}^{\mathit{plog}}_1$ we proceed in the same way as in the proof of Theorem \ref{pedsoplog}[Part a] with the only difference that all states are universal. Let $\phi = \forall X_1^{r_1,\log^{k_1}} \dots \exists X_m^{r_m,\log^{k_m}} \varphi$, where $\varphi$ is a first-order formula with the restrictions given in the definition of $\mathrm{SO}^{\mathit{plog}}$. We first determine all possible values for the second-order variables $X_i^{r_i,\log^{k_i}}$. Any combination of such values determines a branch in the computation tree of \textbf{M}, and for each such branch the machine has to checks $\varphi$. The argument that these checks can be done in poly-logarithmic time is the same as in the proof of Theorem \ref{pedsoplog}. Then by definition of the complexity classes $\mathrm{ATIME}^{op}[\log^k n,m]$ and the definition of acceptance for alternating Turing machines the machine \textbf{M} evaluates $\phi$ in poly-logarithmic time.

In order to show the inverse, i.e. $\tilde{\Pi}^{\mathit{plog}}_1 \subseteq \Pi^{\mathit{plog}}_1$, we exploit that the given random access alternating Turing machine has only universal states and thus all branches in its computation tree must lead to an accepting state. Consequently, the same construction of a formula $\phi$ as in the proof of Theorem \ref{pedsoplog}[Part b] can be used with the only difference that all second-order existential quantifiers have to be turned into universal ones. Then the result follows in the same way as in the proof of Theorem~\ref{pedsoplog}.
\end{proof}

Finally, we get that there is a one-to-one correspondence between the expressive power of the quantifier prefix classes of $\mathrm{SO}^{\mathit{plog}}$ and the levels of the polylog-time hierarchy. 

\begin{theorem}\label{pedsoplog3}

Over ordered structures with successor relation, $\mathrm{BIT}$ and constants for $\log n$, the minimum, second and maximum elements, $\Sigma^{\mathit{plog}}_m$ captures $\tilde{\Sigma}^{\mathit{plog}}_m$ and $\Pi^{\mathit{plog}}_m$ captures $\tilde{\Pi}^{\mathit{plog}}_m$ for all $m \ge 1$.

\end{theorem}

\begin{proof} (sketch)
We proceed by induction, where the grounding cases for $m=1$ are given by Theorems \ref{pedsoplog} and \ref{pedsoplog2}. For the inclusions $\Sigma^{\mathit{plog}}_m \subseteq \tilde{\Sigma}^{\mathit{plog}}_m$ and $\Pi^{\mathit{plog}}_m \subseteq \tilde{\Pi}^{\mathit{plog}}_m$ we have to guess (or take all) values for the second-order variables in the leading block of existential (or universal, respectively) quantifiers, which is done with existential (or universal, respectively) states. For the checking of the subformula in $\Pi^{\mathit{plog}}_{m-1}$ (or in $\Sigma^{\mathit{plog}}_{m-1}$, respectively) we have to switch to a universal (existential) state and apply the induction hypothesis for $m-1$.

Conversely, we consider the computation tree of the given alternating Turing machine \textbf{M} and construct a formulae as in the proofs of Theorems \ref{pedsoplog} and \ref{pedsoplog2} exploiting that for each switch of state from existential to universal (or the other way round) the corresponding submachine can by induction be characterised by a formula in $\Pi^{\mathit{plog}}_{m-1}$ or $\Sigma^{\mathit{plog}}_{m-1}$, respectively.
\end{proof}

The following corollary is a straightforward consequence of Theorem~\ref{pedsoplog3}.
\begin{corollary}

Over ordered structures with successor relation, $\mathrm{BIT}$ and constants for $\log n$, the minimum, second and maximum elements, $\mathrm{SO}^{\mathit{plog}}$ captures the polylog-time hierarchy PLH.

\end{corollary}

\section{On Constant Depth Quasipolynomial Size Circuits and Restricted Second-Order Logics}\label{barrington}

For this section we assume the reader has a basic understanding of circuit complexity (\cite{Immerman99} is a good reference for the subject). We consider a circuit as a connected acyclic digraph with arbitrary number of input nodes and exactly one output node. As in \cite{Barr92} we define  $qAC^0$ as the class of $\mathrm{DTIME} [\log^{O(1)} n]$ $\mathrm{DCL}$ uniform families of Boolean circuits of unbounded fan-in, Size $2^{\log^{O(1)} n}$ and Depth $O(1)$.

If $\mathcal{C}$ is a family of circuits, we consider that $(h,t,g,z_n) \in \mathrm{DCL(\mathcal{C})}$ iff the gate with number $h$ is of type $t$ and the gate with number $g$ is a child of gate $h$, and $z_n$ is an arbitrary binary string of length $n$, if the type is $\vee$, $\wedge$, or $\neg$. If the type is $x$, then $h$ is an input gate that corresponds to bit $g$ of the input. $\mathrm{bin}(\mathbf{A})$, of length $n$,  is the binary encoding of the input structure $\mathbf{A}$ on which $M_{\mathcal{C}}$ will compute the query (see~\cite{Immerman99}). 

  Further, for every $m \geq 1$ we define $qAC^0_{m}$ as the subclass of $qAC^0$ of the families of circuits in $qAC^0$ where the path from  an input gate  to the output gate  with the maximum number of alternated gates of unbounded fan-in of type $\mathrm{AND}$ and $\mathrm{OR}$ in the circuits is $m$.
Following \cite{barrington:jcss1990} we assume that in all the circuits in the family the $\mathrm{NOT}$  gates can occur only at the second level from the left (i.e., immediately following input gates), the gates of unbounded fan-in at any given depth are all of the same type, the layers of such gates alternate in the two types, and the inputs to a gate in a given layer are always outputs of a gate in the previous layer.
Besides the $m$ alternated layers of only gates of unbounded fan-in, we have in each circuit and to the left of those layers a region of the circuit with only $\mathrm{AND}$ and $\mathrm{OR}$ gates of fan-in $2$, with an arbitrary layout, and to the left of that region a layer of some possible $\mathrm{NOT}$ gates and then a layer with the $n$ input gates.
It is straightforward to transform any arbitrary $qAC^0$ circuit into an equivalent one that satisfies such restrictions.
We denote as $\exists qAC^0_{m}$ ($\forall qAC^0_{m}$) the subclass of $qAC^0_{m}$ where the output gate is of type $\mathrm{OR}$ ($\mathrm{AND}$).

The logic $\mathrm{SO}^{\mathit{plog}}$ is closely related to a restricted second-order logic defined by David A. Mix Barrington in~\cite{Barr92}. The logic in~\cite{Barr92} is defined by extending first-order logic with a second-order quantifier $Q_{f}$ for each $f \in \mathcal{F}$ which range over relations on the sub-domain $\{1, \ldots, \log n\}$, where $n$ is the size of the interpreting structure. The case related to our results is when $\mathcal{F} = \{\mathrm{OR}, \mathrm{AND}\}$, which gives raise to restricted existential and universal second-order quantifiers. In the following, we denote the logic obtained when $\mathcal{F} = \{\mathrm{OR}, \mathrm{AND}\}$ as $\mathrm{SO}^b$. It turns out that $qAC^0$ coincides with the class of Boolean queries expressible in $\mathrm{SO}^{b}$ ---the result in \cite{Barr92} is actually more general, allowing any set of Boolean functions  $\mathcal{F}$ of $n^{{O}(1)}$ inputs complying with a padding property and containing the functions $\mathrm{OR}$ and $\mathrm{AND}$.

When considered as a query language, we believe that $\mathrm{SO}^{\mathit{plog}}$ is better suited than $\mathrm{SO}^b$, since it is in general more natural and less cumbersome to define queries in $\mathrm{SO}^{\mathit{plog}}$ than in $\mathrm{SO}^b$. Take for instance the $\mathrm{SO}^{\mathit{plog}}$ sentence in Example~\ref{polylogclique} which expresses a poly-logarithmically bounded version of the clique NP-complete problem. It is not possible to simply and literally express in $\mathrm{SO}^b$ (as we do in $\mathrm{SO}^{\mathit{plog}}$) that there is a set $S$ of arbitrary nodes of $G$ (where $S$ is of size $\lceil \log n\rceil^k$) such that the sub-graph induced by $S$ in $G$ is a clique. Instead in $\mathrm{SO}^b$ we would need to define a set of arbitrary binary numbers, which would need to be encoded into a relation of arity $k+2$ defined on the sub-domain $\{1, \ldots, \log n\}$, and then use $\mathrm{BIT}$ to check whether the nodes of $G$ corresponding to these binary numbers induce a sub-graph of $G$ which is a clique.

There is a well known result (\cite{Immerman99}, Theorem~5.22) which shows that the class of first-order uniform families of Boolean circuits of unbounded fan-in, size $n^{{O}(1)}$ and depth ${O}(1)$, coincides with the class of languages $\mathrm{ATIME}[\log n, {O}(1)]$ that are accepted by random-access alternating Turing machines that make at most $\log n$ steps and at most ${O}(1)$ alternations between existential and universal states. The intuitive idea is that as alternating Turing machines have bounded fan-out in their computation trees, to implement an AND (OR) gate of unbounded fan-in, a full balanced tree of depth logarithmic in the size of the circuits, of universal (existential) states is needed. Then it appears as natural that $qAC^0$ coincides with the whole poly-logarithmic time hierarchy $\mathrm{PLH}$ as defined in this paper, since $\mathrm{PLH} = \mathrm{ATIME}[(\log n)^{{O}(1)}, {O}(1)]$ and $(\log n)^{{O}(1)}$ is the logarithm of the size $2^{(\log n)^{{O}(1)}}$ of the circuits in $qAC^0$.   

Therefore the fact that $\mathrm{SO}^{\mathit{plog}}$ captures the whole class $\mathrm{PLH}$ could also be seen as a corollary of Barrington's theorem in~\cite{Barr92} (see Section~3, page~89). This however does \emph{not} applies to our main results, i.e., the capture of $\npolylog$ by the existential fragment of $\mathrm{SO}^{\mathit{plog}}$ and the one-to-one correspondence between the quantifier prefix classes of $\mathrm{SO}^{\mathit{plog}}$ and the corresponding levels of $\mathrm{PLH}$. The critical difference between Barrington's $\mathrm{SO}^b$ logic and $\mathrm{SO}^{\mathit{plog}}$ is that we impose a restriction in the first-order logic sub-formulae, so that the universal first-order quantifier is only allowed to range over sub-domains of polylog size. This is a key feature since otherwise the first-order sub-formulae of the $\Sigma^{\mathit{plog}}_m$ (and $\Pi^{\mathit{plog}}_m$) fragments of $\mathrm{SO}^{\mathit{plog}}$ would need at least linear time to be evaluated. Of course, Barrington does not need to define such constraint because he always speaks of the whole class $\mathrm{PLH}$, and we show indeed that for every first-order logic formula there is an equivalent $\mathrm{SO}^{\mathit{plog}}$ formula. 

We do not know yet whether an exact correspondence between the levels $\tilde{\Sigma}_m^{\mathit{plog}}$ of the polylog-time hierarchy PLH and ``natural'' sub-classes of families of circuits in $\mathit{qAC}^0$ can be established. So far we have proven the following two lemmas, but it is open whether their converses hold or not.

\begin{lemma}\label{a}
For all $m \ge 1$ we have that $\exists qAC^0_{m}$ $\subseteq$ $\tilde{\Sigma}^{\mathit{plog}}_m$, and $\forall qAC^0_{m}$ $\subseteq$ $\tilde{\Sigma}^{\mathit{plog}}_{m+1}$.


\end{lemma}

\begin{proof}
Let $\mathcal{C}$ be a family of circuits in $\exists qAC^0_{m}$ with the uniformity conditions given above. We build an alternating Turing machine $M_{\mathcal{C}}$ that computes the query computed by $\mathcal{C}$. We have a deterministic Turing machine that decides the $\mathrm{DCL}$ of $\mathcal{C}$ in time $\log^{c} n$, for some constant $c$. Then  for any given pair of gate numbers $g, h$, gate type $t$, and arbitrary string of $n$ bits $z_{n}$, we can deterministically check both $(h,t,g,z_n) \in \mathrm{DCL(\mathcal{C})}$ and $(h,t,g,z_n) \not\in \mathrm{DCL(\mathcal{C})}$. To be able to have a string $z_n$ of size $n$, we allow the input tape of $M_{\mathcal{C}}$ to be read/write, so that to compute those queries we write   $h,t,g,$ to the left of the input $\mathrm{bin}(\mathbf{A})$ in the input tape of the machine. Note that each gate number is $O(\log^{c'} n)$ bits long for some constant $c'$, and the type $t$ is $2$ bits long (which encodes the type in $\{\mathrm{AND},$ $\mathrm{OR},$ $\mathrm{NOT},$ $x\}$, where $x$ indicates that the gate is an input gate).

First $M_{\mathcal{C}}$ computes the size $n$ of the input $\mathrm{bin}(\mathbf{A})$, which can done in logarithmic time (see \cite{barrington:jcss1990}).

Corresponding to the $m$ alternated layers of  gates of unbounded fan-in of type $\mathrm{AND}$ and $\mathrm{OR}$, where the last layer consists of one single $\mathrm{OR}$ gate which is the output gate, we will have in $M_{\mathcal{C}}$ $m$ alternated blocks of existential and universal states, which will be executed in the opposite direction to the edge relation in the circuit, so that the first block, which is existential will correspond to layer $m$ in $\mathcal{C}$. Each such block  takes time $O(\log^{c'} n)$. In the following we will consider the layers from right to left.
\begin{itemize}
 \item
In the first block, which is \textbf{existential}, $M_{\mathcal{C}}$ guesses gate numbers $g_{o}, h_{2}$, and checks whether $(h_{2}, \wedge, g_{o}, z_n)$ $\in \mathrm{DCL(\mathcal{C})}$. If it is true, then $M_{\mathcal{C}}$ writes in a work tape the sequence $\langle g_{o}, h_{2}\rangle$.
Note that $h_{2}$ is of type $\wedge$. Then it checks
whether $(h_{2}, \wedge, g_{o},  z_n)$ $\not\in \mathrm{DCL(\mathcal{C})}$. If it is true then $M_{\mathcal{C}}$ rejects.

 \item
The second block is \textbf{universal}, and has two stages. In the first stage, $M_{\mathcal{C}}$ checks whether $g_{o}$ is the output gate. To that end it guesses a gate number $u$, and checks whether $(g_{o}, \vee, u,  z_n)$ $\in \mathrm{DCL(\mathcal{C})}$. If it is true, then $M_{\mathcal{C}}$ rejects.

In the second stage, $M_{\mathcal{C}}$ checks the inputs to gate $h_{2}$. It guesses a gate number $h_{3}$, and checks whether $(h_{3}, \vee,  h_{2}, z_n)$ $\in \mathrm{DCL(\mathcal{C})}$, in which case it adds $h_{3}$ at the right end of the sequence in the work tape. It then checks whether $(h_{3}, \vee, h_{2},  z_n)$ $\not\in \mathrm{DCL(\mathcal{C})}$, in which case it accepts.

 \item
The third block is \textbf{existential}. $M_{\mathcal{C}}$ checks the inputs to gate $h_{3}$. It guesses a gate number $h_{4}$, and checks whether
$(h_{4}, \wedge, h_{3}, z_n)$ $\in \mathrm{DCL(\mathcal{C})}$, in which case it adds $h_{4}$ at the right end of the sequence in the work tape. It then checks whether
$(h_{4}, \wedge, h_{3}, z_n)$ $\not\in \mathrm{DCL(\mathcal{C})}$, in which case it rejects.

 \item
Following the same alternating pattern, the $(m-1)$-th block will be existential or universal depending on the type of the gates in $\mathcal{C}$ at the $(m- 1)$-th layer.
Note that, in our progression from the output gate towards the input gates (right to left), the parents of $h_{m}$ (which is the gate number guessed at the $(m-1)$-th block of $M_{\mathcal{C}}$) are the first gates in the region of $\mathcal{C}$ of the arbitrary layout of only gates with bounded fan-in. As the depth of each circuit in the family $\mathcal{C}$ is constant, say it is $w$ for all the circuits in the family, in the $m$-th block $M_{\mathcal{C}}$ can guess the \textit{whole sub-circuit} of that region. Then it guesses $w$ gate  numbers and checks that they form exactly the layout of that region of the circuit. Once  $M_{\mathcal{C}}$ has guessed that layout it can work deterministically to evaluate it, up to the input gates, which takes time $O(1)$.

If the $(m-1)$-th block is \textbf{universal}, then if the guessed $w$ gate numbers do not form the correct layout, $M_{\mathcal{C}}$ accepts.
If the $(m-1)$-th block is \textbf{existential}, and the guessed $w$ gate numbers do not form the correct layout, then $M_{\mathcal{C}}$ rejects.


\end{itemize}


Note that if the family $\mathcal{C}$ is in $\forall qAC^0_{m}$, the first block of states is still existential, to guess the output gate. Then it works as in the  case of $\exists qAC^0_{m}$. 
\end{proof}

\begin{lemma}\label{b}
Let $t,k \ge 1$ and $\psi \in {\Sigma}^{1,\mathit{plog}}_t$ with first-order sub-formula $\varphi \in {\Sigma}^{0}_k$, and whose vocabulary includes the $\mathrm{BIT}$ predicate.
Then there is a family $\mathcal{C}_{\psi}$ of Boolean circuits in $\exists qAC^0_{t+k}$ that computes the Boolean query expressed by $\psi$.

\end{lemma}

\begin{proof}

We essentially follow the sketch of the proof of the theorem in Section 3, page 89 of \cite{Barr92}, where it shows that for a given $\psi \in $ $\mathrm{SO}^{\mathit{b}}$ there is an equivalent circuit family $\mathcal{C}_{\psi}$ in $qAC^0$. But we use a simple strategy to define a layout of the circuits which will preserve the number of alternated blocks of quantifiers in $\psi$ and $\varphi$.

 Starting from the canonical $AC^0$ circuit corresponding to a first-order formula in quantifier prefix normal form, as in Theorem 9.1 in \cite{barrington:jcss1990}, we follow the same idea extending such circuit $\mathcal{C}_{\varphi}$ for $\varphi \in {\Sigma}^{0}_k$, to a canonical circuit $\mathcal{C}_{\psi}$ in $\exists qAC^0$
for $\psi \in {\Sigma}^{1,\mathit{plog}}_t$, in such a way that $\mathcal{C}_{\psi}$ is in $\exists qAC^0_{t+k}$.

The layout of $\mathcal{C}_{\varphi}$ basically follows from left to right the opposite order of the formula $\varphi$, so that the output gate is an unbounded fan-in $\vee$ gate that corresponds to the first-order quantifier $\exists_1$, the inputs to that gate are the outputs of a layer of unbounded fan-in $\wedge$ gates that correspond to the first-order quantifier $\forall_2$, and so on. To the left of the leftmost layer of unbounded fan-in gates corresponding to the quantifier $Q_k$, there is a constant size, constant depth region of the circuit which corresponds to the quantifier free sub-formula of $\varphi$. This part has the input gates, constants, $\mathrm{NOT}$ gates, and  $\mathrm{AND}$ and $\mathrm{OR}$ gates of fan-in $2$.

In a similar way, we extend $\mathcal{C}_{\varphi}$ to the right, to get $\mathcal{C}_{\psi}$. To that end, to the right of the first-order quantifier $\exists_1$ (which in $\mathcal{C}_{\psi}$ becomes a \textit{layer} of unbounded fan-in $\vee$ gates), we will have in $\mathcal{C}_{\psi}$ one layer of  gates of unbounded fan-in for each $SO$ quantifier: of $\vee$ gates for existential $SO$ quantifiers, and $\wedge$ gates for universal $SO$ quantifiers. The first added layer corresponds to the $SO$ quantifier $Q_t$, and the rightmost layer will correspond the $SO$ quantifier $\exists_1$, which is a layer of one single gate, that becomes the new output gate.

Clearly, by following the construction above we get a family of circuits in $\exists qAC^0_{t+k}$.

To build the Turing machine $M_{\mathcal{C}}$ $\in\mathrm{DTIME} [\log^{O(1)} n]\;$ that decides the language $\mathrm{DCL}(\mathcal{C})$, we use the same kind of encoding sketched in \cite{Barr92} (which in turn is an extension of the one used in \cite{barrington:jcss1990} for $AC^{0}$).
In the number of the gates we encode all the information that we need to decide the language, while still keeping its length polylogarithmic.
Each such gate number will have different sections: i) type of the gate; ii) a sequence of $k$ fields of polylogarithmic size each to hold the values of the first-order variables; iii) a sequence of $t$ fields of polylogarithmic size each to hold the values of of the $SO$ variables; iv) a constant size field for the code of the $\mathrm{NOT}$, $\mathrm{OR}$ and $\mathrm{AND}$ gates of fan-in $1, 2, 2$, respectively, in the region of $\mathcal{C}_{\psi}$ corresponding to the quantifier-free part of the formula $\varphi$; v) a logarithmic size field for the bit number that corresponds to an input gate, vi) a polylogarithmic size field for the number of the gate whose output is the left (or only) input to the gate; and vii) idem for the right input.

The idea is that for any given gate, its number will hold the values of all the first-order and $SO$ variables that are bounded in ${\psi}$, in the position of the formula that corresponds, by the construction above, to that gate, or zeroes if the variable is free. Note that each sequence of all those values uniquely define a path in $\mathcal{C}_{\psi}$ from the output gate to the given gate. For the gates which correspond to quantifiers, and for those in the quantifier free part that are parents of them, the encoding allows to easily compute the number of their child gates.
The gates for the quantifier free part hold in their numbers (iv) a number which uniquely identifies that gate in that region of the circuit for a particular branch in $\mathcal{C}_{\psi}$ which is given by the values of the bounded variables. Note that the layout of each such branch of the circuit is constant and hence \textit{stored} in the transition function of $M_{\mathcal{C}}$, and it can evaluate that sub circuit in polylogarithmic time.
Note that the predicate $\mathrm{BIT}(i,j)$ can be evaluated by $M_{\mathcal{C}}$ by counting in binary in a work tape up to $j$ and then looking at its bit $i$.

In this way, clearly $M_{\mathcal{C}}$ decides $\mathrm{DCL}(\mathcal{C})$ in time $\log^{O(1)} n$. 
\end{proof}

\section{Conclusions}\label{sec:schluss}
We investigated $\mathrm{SO}^{\mathit{plog}}$, a restriction of second-order logic, where second-order quantification ranges over relations of poly-logarithmic size and first-order quantification is restricted to the existential fragment of first-order logic plus universal quantification over variables in the scope of a second-order variable. In this logic we defined the poly-logarithmic hierarchy PLH using fragments $\Sigma^{\mathit{plog}}_m$ (and $\Pi^{\mathit{plog}}_m$) defined by formulae with alternating blocks of existential and universal second-order quantifiers in quantifier prefix normal form. We show that the existential fragment $\Sigma^{\mathit{plog}}_1$ captures $\npolylog$, i.e. the class of Boolean queries that can be accepted by a non-deterministic Turing machine with random access to the input in time $O(\log^k n)$ for some $k \ge 0$. In general, $\Sigma^{\mathit{plog}}_m$ captures the class of Boolean queries that can be accepted by an alternating Turing machine with random access to the input in time $O(\log^k n)$ for some $k \ge 0$ with at most $m$ alternations between existential and universal states. Thus, PLH is captured by $\mathrm{SO}^{\mathit{plog}}$. 

For the proofs the restriction of first-order quantification is essential, but it implies that we do not have closure under negation. As a consequence we do not have a characterisation of the classes $\text{co-}\Sigma^{\mathit{plog}}_m$ and $\text{co-}\Pi^{\mathit{plog}}_m$. These consitute open problems. Furthermore, PLH resides in the complexity class PolyLogSpace, which is known to be different from $\mathrm{P}$, but it is conjectured that PolyLogSpace and $\mathrm{P}$ are incomparable. Whether the inclusion of PLH in PolyLogSpace is strict is another open problem.

The theory developed in this article and its proofs make intensive use of alternating Turing machines with random access to the input. We observe that it appears awkward to talk about poly-logarithmic time complexity, when actually an unbounded number of computation branches have to be exploited in parallel. It appears more natural to refer directly to a computation model that involves directly unbounded parallelism such as Abstract State Machines that have already been explored in connection with the investigation of choiceless polynomial time \cite{blass:apal1999}. We also observe that a lot of the technical difficulties in the proofs result from the binary encodings that are required in order to make logical structures accessible for Turing machines. The question is, whether a different, more abstract treatment would help to simplify the technically complicated proofs. These more general questions provide further invitations for future research.

\bibliographystyle{plain}

\bibliography{SOpolylog}

\begin{thebibliography}{10}

\bibitem{babai:stoc2016}
L\'{a}szl\'{o} Babai.
\newblock Graph isomorphism in quasipolynomial time.
\newblock In {\em Proceedings of the forty-eighth annual ACM symposium on
  Theory of Computing (STOC 2016)}, pages 684--697, 2016.

\bibitem{blass:apal1999}
Andreas Blass, Yuri Gurevich, and Saharon Shelah.
\newblock Choiceless polynomial time.
\newblock {\em Ann. Pure Appl. Logic}, 100(1-3):141--187, 1999.

\bibitem{BruynoogheB0CPJ15}
Maurice Bruynooghe, Hendrik Blockeel, Bart Bogaerts, Broes~De Cat, Stef~De
  Pooter, Joachim Jansen, Anthony Labarre, Jan Ramon, Marc Denecker, and Sicco
  Verwer.
\newblock Predicate logic as a modeling language: modeling and solving some
  machine learning and data mining problems with \emph{IDP3}.
\newblock {\em {TPLP}}, 15(6):783--817, 2015.

\bibitem{Cook_71}
Stephen~A. Cook.
\newblock The complexity of theorem-proving procedures.
\newblock In {\em Proceedings of the Third Annual ACM Symposium on Theory of
  Computing}, STOC '71, pages 151--158, New York, NY, USA, 1971. ACM.

\bibitem{Dawar98}
Anuj Dawar.
\newblock A restricted second order logic for finite structures.
\newblock {\em Inf. Comput.}, 143(2):154--174, 1998.

\bibitem{DurandLS98}
Arnaud Durand, Clemens Lautemann, and Thomas Schwentick.
\newblock Subclasses of binary {NP}.
\newblock {\em J. Log. Comput.}, 8(2):189--207, 1998.

\bibitem{EiterGS01}
Thomas Eiter, Georg Gottlob, and Thomas Schwentick.
\newblock Second-order logic over strings: Regular and non-regular fragments.
\newblock In Werner Kuich, Grzeg orz Rozenberg, and Arto Salomaa, editors, {\em
  Developments in Language Theory}, volume 2295 of {\em Lecture Notes in
  Computer Science}, pages 37--56. Springer, 2001.

\bibitem{EiterGG00}
Thomas Eiter, Yuri Gurevich, and Georg Gottlob.
\newblock Existential second-order logic over strings.
\newblock {\em J. ACM}, 47(1):77--131, 2000.

\bibitem{fagin:1973}
Ronald Fagin.
\newblock {\em Contributions to Model Theory of Finite Structures}.
\newblock PhD thesis, U. C. Berkeley, 1973.

\bibitem{FerrarottiGST08}
Flavio Ferrarotti, Sen{\'{e}}n Gonz{\'{a}}lez, Klaus-Dieter Schewe, and
  Jos{\'{e}}~Mar\'{\i}a {Turull Torres}.
\newblock The polylog-time hierarchy captured by restricted second-order logic.
\newblock In {\em Post-Proceedings of the 20th International Symposium on
  Symbolic and Numeric Algorithms for Scientific Computing (To appear)}. IEEE
  Computer Society, 2019.

\bibitem{GottlobKS04}
Georg Gottlob, Phokion~G. Kolaitis, and Thomas Schwentick.
\newblock Existential second-order logic over graphs: Charting the tractability
  frontier.
\newblock {\em J. ACM}, 51(2):312--362, 2004.

\bibitem{Gradel92}
Erich Gr{\"a}del.
\newblock Capturing complexity classes by fragments of second-order logic.
\newblock {\em Theor. Comput. Sci.}, 101(1):35--57, 1992.

\bibitem{graedel:eatcs2007}
Erich Gr{\"{a}}del, Phokion~G. Kolaitis, Leonid Libkin, Maarten Marx, Joel
  Spencer, Moshe~Y. Vardi, Yde Venema, and Scott Weinstein.
\newblock {\em Finite Model Theory and Its Applications}.
\newblock Texts in Theoretical Computer Science. An {EATCS} Series. Springer,
  2007.

\bibitem{GrossoT10}
Alejandro~L. Grosso and Jos{\'e}~Mar{\'{\i}}a {Turull Torres}.
\newblock A second-order logic in which variables range over relations with
  complete first-order types.
\newblock In Sergio~F. Ochoa, Federico Meza, Domingo Mery, and Claudio
  Cubillos, editors, {\em SCCC}, pages 270--279. IEEE Computer Society, 2010.

\bibitem{HouCD10}
Ping Hou, Broes~De Cat, and Marc Denecker.
\newblock {FO(FD):} extending classical logic with rule-based fixpoint
  definitions.
\newblock {\em {TPLP}}, 10(4-6):581--596, 2010.

\bibitem{Immerman99}
Neil Immerman.
\newblock {\em Descriptive complexity}.
\newblock Graduate texts in computer science. Springer, 1999.

\bibitem{LautemannST94}
Clemens Lautemann, Thomas Schwentick, and Denis Th{\'e}rien.
\newblock Logics for context-free languages.
\newblock In Leszek Pacholski and Jerzy Tiuryn, editors, {\em CSL}, volume 933
  of {\em Lecture Notes in Computer Science}, pages 205--216. Springer, 1994.

\bibitem{Libkin04}
Leonid Libkin.
\newblock {\em Elements of Finite Model Theory}.
\newblock Springer, 2004.

\bibitem{Barr92}
David~A. {Mix Barrington}.
\newblock Quasipolynomial size circuit classes.
\newblock In {\em Proceedings of the Seventh Annual Structure in Complexity
  Theory Conference, Boston, Massachusetts, USA, June 22-25, 1992}, pages
  86--93. {IEEE} Computer Society, 1992.

\bibitem{barrington:jcss1990}
David~A. {Mix Barrington}, Neil Immerman, and Howard Straubing.
\newblock On uniformity within {NC}$^1$.
\newblock {\em J. Comput. Syst. Sci.}, 41(3):274--306, 1990.

\bibitem{Stockmeyer76}
Larry~J. Stockmeyer.
\newblock The polynomial-time hierarchy.
\newblock {\em Theor. Comput. Sci.}, 3(1):1--22, 1976.

\end{thebibliography}

\end{document}